\documentclass[letterpaper,12pt]{article}
\usepackage[latin1]{inputenc}
\usepackage{amsmath, graphicx, layout, verbatim,rotating}
\usepackage{setspace}
\usepackage[colorlinks=true,urlcolor=blue,citecolor=blue	,linkcolor=blue]{hyperref}
\usepackage{fancyhdr}
\usepackage{float}
\usepackage{color}
\usepackage{multirow}
\usepackage{natbib}
\bibpunct[; ]{(}{)}{;}{a}{,}{;}
\usepackage{verbatim}
\usepackage{subfig}
\usepackage{amsthm}
\usepackage{amsfonts}
\usepackage{algorithm}
\usepackage{algpseudocode}
\usepackage{amssymb}

  \usepackage{enumitem}
\setlist{nolistsep}

\renewcommand{\algorithmicrequire}{\textbf{\small Input:}}
\renewcommand{\algorithmicensure}{\textbf{\small Output:}}

\usepackage{setspace}

\def\X{***ATT***}
\newcommand{\red}[1]{\textbf{\color{red} #1}}

\renewcommand{\vec}[1]{\mathbf{#1}}
\renewcommand{\comment}[1]{}

\def\I{{\mathbb I}}
\def\P{{\mathbb P}}
\def\E{{\mathbb E}}
\def\V{{\mathbb V}}
\def\X{{\vec{X}}}
\def\x{{\vec{x}}}

\def\itemrange#1{%
\addtocounter{enumi}{1}%
\edef\labelenumi{\theenumi--\noexpand\theenumi.}%
\addtocounter{enumi}{-1}%
\addtocounter{enumi}{#1}%
\item
\def\labelenumi{\theenumi}}

\newtheorem{thm}{Theorem}

\newtheorem{Cor}{Corollary}
\newtheorem{prop}{Proposition}
\newtheorem{Lemma}{Lemma}
\newtheorem{remark}{Remark}
\newtheorem{Assumption}{Assumption}

 \graphicspath{{../../}}
\usepackage{xpatch}
\xpatchcmd{\algorithmic}{\setcounter}{\algorithmicfont\setcounter}{}{}
\providecommand{\algorithmicfont}{\small}

\theoremstyle{definition}
\newtheorem{ex}{Example}

\providecommand{\keywords}[1]{\textbf{\textbf{Key Words:}} #1}

\begin{document}


\title{ABC-CDE: Towards Approximate Bayesian Computation with Complex High-Dimensional Data and Limited Simulations}
\date{}
\author{Rafael Izbicki\thanks{Department of Statistics, Federal University of S\~ao Carlos, Brazil.}, \ Ann B. Lee\thanks{Department of Statistics \& Data Science, Carnegie Mellon University, USA.}\  and Taylor Pospisil\thanks{Department of Statistics \& Data Science, Carnegie Mellon University, USA. }}

\maketitle 

\begin{abstract} 
Approximate Bayesian Computation (ABC) is typically used when the likelihood is either unavailable or intractable but where data can be simulated under different parameter settings using a forward model. Despite the recent interest in ABC, high-dimensional data and costly simulations still remain a bottleneck in some applications. There is also no consensus as to how to best assess the performance of such methods without knowing the true posterior. We show how a nonparametric conditional density estimation (CDE)
framework,  which we refer to as ABC-CDE, help address three nontrivial challenges in ABC: (i) how to efficiently estimate the posterior
distribution with limited simulations and different types of data,
(ii) how to tune and compare the performance of ABC and related methods in estimating the posterior itself, rather than just certain properties of the density,
and (iii) how to
efficiently choose among a large set of summary statistics based on a CDE surrogate loss.
We provide theoretical and empirical evidence that justify ABC-CDE procedures that {\em directly} estimate and assess the posterior  based on an initial ABC sample, and we describe settings where standard ABC and regression-based approaches are inadequate.

\end{abstract}

  \section{Introduction}

For many statistical inference problems in the sciences
the relationship between the parameters of interest and observable data is complicated,
but it is possible to simulate realistic data according to some model;
see  \citet{beaumont2010approximate,estoup2012estimation}
for examples in genetics, and  \citet{cameron2012approximate,weyant2013likelihood}
for examples in astronomy. 
In such situations,
the complexity of the data generation process
often prevents the derivation of a sufficiently accurate analytical form for the likelihood function.
One cannot use standard Bayesian tools as no analytical form for the posterior distribution is available.
Nevertheless one can estimate $f(\theta|\x)$, the posterior distribution of the parameters $\theta \in \Theta$ given data $\x \in \mathcal{X}$, by taking advantage of the fact that it is possible to forward simulate data $\x$ under different settings of the parameters $\theta$.
Problems of this type have motivated recent interest in methods of {\it likelihood-free inference}, 
which includes methods of {\it Approximate Bayesian Computation} (ABC; \citealt{marin2012approximate}) 

Despite the recent surge of approximate Bayesian methods, 
several challenges still remain. 
In this work,  we present a \emph{conditional density estimation} (CDE)
framework and  {\em a surrogate loss function for CDE} that address the following three problems:
\begin{enumerate}[label=(\roman*)]
	\item   how to efficiently estimate the posterior density $f(\theta|\x_o)$, where $\x_o$ is the observed sample; in particular, in settings with complex, high-dimensional data  and costly simulations,
		\item  how to  choose tuning parameters and  compare the performance of ABC and related methods based on simulations and observed data only; that is, without knowing the true posterior distribution, 	and
	\item how to best choose
	summary statistics for ABC and related methods when
	given a very large number of candidate summary statistics.
\end{enumerate}


{\bf Existing Methodology.}
There is an extensive literature on ABC methods; we refer the reader to \citet{marin2012approximate,prangle2014diagnostic}
and references therein for a review.
The connection between ABC and CDE has been noted by others;
in fact, ABC itself can be viewed
as a hybrid between
nearest neighbors and kernel density estimators \citep{blum2010approximate,biau2015new}.   As Biau et al. point out, the fundamental problem from a practical perspective is how to select the parameters in ABC methods in the absence of a priori information regarding the posterior $f(\theta|\x_o)$. Nearest neighbors and kernel density estimators are also  known to perform poorly in  settings with a large amount of summary statistics
\citep{blum2010approximate}, and they are difficult to adapt to different data types (e.g.,  mixed discrete-continuous statistics and functional data).  Few works attempt to use other CDE methods to estimate posterior distributions. At the time of submission of this paper, the only works in this direction are \citet{papamakarios2016fast} and \citet{lueckmann2017flexible}, which are based on conditional neural density estimation, \citet{fan2013approximate}  and \citet{li2015extending}, which use a mixture of Gaussian copulas to estimate
the likelihood function, and 
\citet{marin2017abc}, which suggests random forests for quantile estimation.

\comment{
show that improved estimation
can be obtained by relying upon the success of the novel
conditional density estimation framework introduced by \citet{IzbickiLeeFlexCode}, named FlexCode ({\em Flex}ible nonparametric {\em co}nditional {\em d}ensity {\em e}stimation via regression), which is a powerful
tool for converting \emph{any} regression estimator into a conditional
density estimator.  It follows that, in the ABC context, FlexCode allows one to
achieve better statistical performance for a smaller number of simulations, and hence, better performance
for smaller computation times.
}

\comment{
\red{TO BE INCLUDED: Ann says: Related to this, we should also mention the literature on ?postsampling regression adjustment?. Strictly speaking, our approach is a ?postprocessing scheme? if one uses the definition of Marin et al (see Stat Comp 2010, Section 5) ?- however, to me, it seems that existing work all relies on some post adjustment *regression* of theta on x (or regression of x on theta) rather than conditional density estimation. In any case, we need to cite the current literature and point out how our work is different.}
}

Although the above mentioned methods utilize specific CDE models to estimate posterior distributions, they do not fully explore other
advantages of a CDE framework;  such as, in methods assessment, in variable selection, 
and in 
 tuning the final estimates
with CDE as a goal (see Sections \ref{sec::comparing}, \ref{sec::summary}
and \ref{sec::beyond_ABC}). 
Summary statistics selection is indeed a 
nontrivial challenge in likelihood-free inference: 
ABC methods 
depend on the choice of statistics and distance function of observables when comparing observed and simulated data, and using the ``wrong'' summary statistics can 
 dramatically affect  their performance. For a general review of dimension reduction methods for ABC, we refer the reader to \citet{blum2013comparative},
who classify current approaches in three classes: (a) best subset selection approaches, (b) projection techniques, and (c)
regularization techniques. Many of these approaches still face either 
significant computational issues or attempt to find good summary statistics for certain characteristics of the posterior rather than the {\em entire} posterior itself.
  For instance,
\citet{creel2016selection} propose a best subset selection of summary statistics based on improving the estimate
of the posterior mean $\E[\theta|\x]$. There are no guarantees however that statistics that
lead to good estimates of $\E[\theta|\x]$ will  
be sufficient for $\theta$ or even yield reasonable estimates of   $f(\theta|\x)$.\footnote{As an example,
if $X_1,\ldots,X_n \sim \mbox{Unif}(\theta,\theta+1)$, the minimal sufficient summary statistic for $\theta$
is $(\min\{X_1,\ldots,X_n\},\max\{X_1,\ldots,X_n\})$.  The optimal statistic for estimating the posterior mean, on the other hand, is $E=\E[\theta|\x]$ \citep[Section 2.3]{fearnhead2012constructing}.}
We will elaborate on this point in Section \ref{sec::comparing}.

   \begin{figure}[H]
   	\centering
   	\includegraphics[page=1,width=6.6in,height=3.3in]{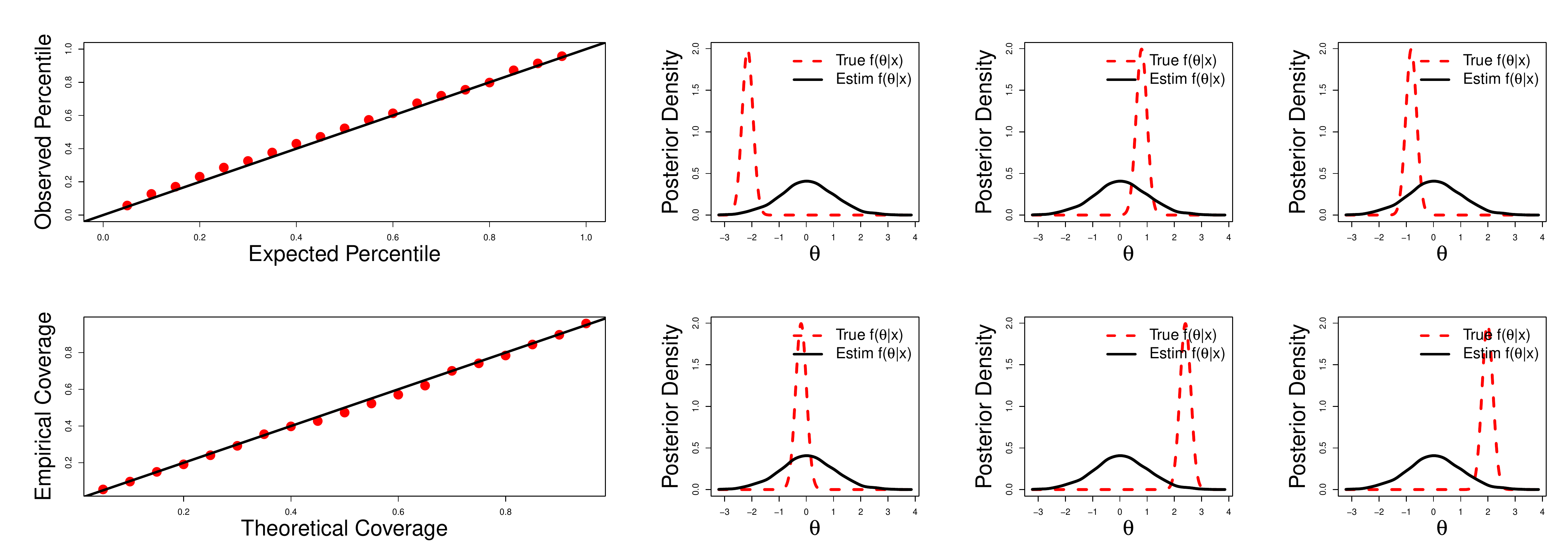} 
   	\vspace{-10mm}
   	\caption{\footnotesize   Limitations of diagnostic tests in conditional density estimation. The PP and coverage plots to the left indicate an excellent fit of $f(\theta \vert x)$ but, as indicated by the examples to the right of a few different values of $x$, the estimated posterior densities (solid black lines) are very far from the true densities (red dashed lines).} 
   	\label{fig::knn_example}
   \end{figure}

Moreover, the current literature on likelihood-free inference lacks methods that allow
one to directly compare the performance of different posterior distribution estimators. 
Given a collection of estimates $\widehat{f}_1(\theta|\x_o), \ldots, \widehat{f}_m(\theta|\x_o)$
(obtained by, e.g., ABC methods with different
tolerance levels, sampling techniques, and so on),  an open problem is to how to select the  
 estimate that is closest to the true posterior density $f(\theta|\x_o)$ for observed data $\x_o$.
Some goodness-of-fit techniques
have been proposed (for example, \citealt{prangle2014diagnostic} compute the goodness of fit based on coverage properties), 
 but although diagnostic tests are useful, they do not capture all aspects of the density estimates. Some density estimates which are not close to the true density can pass all tests \citep{breiman2001statistical,bickel2006tailor}, and the situation is even worse in conditional density estimation.   Figure \ref{fig::knn_example}  shows a toy example where both probability-probability (PP) and coverage plots wrongly indicate an excellent fit,\footnote{See \citet{IzbickiLeeFlexCode} for details on computations of these plots.} but 
 the estimated posterior distributions
    	are far from the true densities; here, $\theta|x\sim  \operatorname{Normal}(x,0.3^2)$ and $X \sim \operatorname{Normal}(0,1)$. 
    	Indeed, standard diagnostic tests will not detect an obvious flaw in conditional density estimates $\widehat{f}(\theta|\x)$ that, as in this example, are equal to the marginal distribution $f(\theta)=\int f(\theta|\x') f(\x') d\x'$ for all  $\x$.


In this paper, we  show how one can improve ABC methods with a novel CDE surrogate loss function (Eq. \ref{eq::surrogateLoss})  that measures how well one estimates the entire posterior distribution ${f}(\theta|\x_o)$; see Section
\ref{sec::comparing} for a discussion of its theoretical properties.  Our proposed method, ABC-CDE, starts with a rough approximation from an ABC sampler and then {\em directly} 
 estimates the conditional density exactly at the point $\x = \x_o$ using a nonparametric conditional density estimator.
 Unlike other ABC post-adjustment techniques  in the literature  (e.g. \citet{beaumont2002approximate} and \citet{blum2010non}), our method is optimized for estimating posteriors, and corrects for changes in the ABC posterior sample beyond the posterior mean and variance. We also present a  general framework (based on CDE) that can handle different types of data (including functional data, mixed variables, structured data, and so on) as well as a larger number of summary statistics. With, for example, FlexCode~\citep{IzbickiLeeFlexCode} 
 one can convert any existing regression estimator to a conditional density estimator. Recent neural mixture density networks  that directly estimate posteriors for complex data \citep{papamakarios2016fast,lueckmann2017flexible}  also fit into this ABC-CDE framework, and are especially promising for image data.
 Hence, with ABC-CDE, we take a different approach to address the curse of dimensionality than in traditional likelihood-free inference methods. In standard ABC,  it is essential to choose (a smaller set of) informative summary statistics to properly measure user-specified distances between observed and simulated data. The main dimension reduction in our framework is implicit in the conditional density estimation and our CDE loss function. Depending of the choice of estimator, we can adapt to different types of sparse structure in the data, and just as in high-dimensional regression, handle a large amount of covariates, even without relying on a {\em prior} dimension reduction of the data and a user-specified distance function of observables. 


Finally, we note that ABC
summary statistic selection and goodness-of-fit 
techniques are typically designed to estimate posterior distributions accurately  
for every sample $\x$. In reality, we often only care about estimates for the particular sample $\x_o$ that is observed, and even if a method produces poor estimates for some $f(\theta|\x')$ it can still produce good estimates for
$f(\theta|\x_o)$. The methods we introduce in this paper take this into consideration, and directly aim at constructing, evaluating and tuning estimators for the posterior density $f(\theta|\x_o)$ at the {\em observed} value $\x_o$.

The organization of the paper is as follows: Section \ref{sec::methods}
describes and presents theoretical results for how a CDE framework and a surrogate loss function address issues (i)--(iii).
 Section \ref{sec::experiments} includes 
simulated experiments 
that demonstrate that our proposed methods  work in practice.
In Section \ref{sec::beyond_ABC}, we revisit CDE in the context of ABC and 
demonstrate how direct estimation of posteriors with CDE and the surrogate loss can replace further iterations with standard ABC. We then end by providing links to  general-purpose CDE software that can be used in likelihood-free inference in different settings. We refer the reader to the Appendix for proofs, a comparison to post-processing regression adjustment methods, and two applications in astronomy.


\section{Methods}
\label{sec::methods}

In this section we propose a CDE framework 
for (i) estimating
  the posterior density (Section \ref{sec::postEst}),
 (ii)  comparing
  the performance of ABC and related methods
  (Section \ref{sec::comparing}), and
  (iii)  choosing optimal
  summary statistics (Section \ref{sec::summary}).

\subsection{Estimating the Posterior Density via CDE}
\label{sec::postEst}

Given a prior distribution $f(\theta)$
and a likelihood function $f(\x|\theta)$,
our goal is to compute  the posterior distribution $f(\theta|\x_o)$, where $\x_o$
is the observed sample. 
We assume we know how to sample from
$f(\x|\theta)$  for a fixed value of $\theta$.

A naive way of estimating $f(\theta|\x_o)$  via CDE methods is to 
first generate an i.i.d.~sample $\mathcal{T} = \{(\theta_1,\X_1), \ldots, (\theta_B,\X_B)\}$ 
by sampling $\theta \sim f(\theta)$ and then $\X \sim f(\x|\theta)$ for each pair. One applies the CDE method of choice to $\mathcal{T}$,  and then simply evaluates the estimated density $\widehat{f}(\theta|\x)$ at $\x=\x_o$. The naive approach however may lead to poor results because some $\x$ are far from  the observed data $\x_o$. To put it differently, standard conditional density estimators are designed
to estimate $f(\theta|\x)$ for every $\x$, but in ABC applications we are only interested in $\x_o$.

To solve this issue, 
one can instead estimate $f(\theta|\x)$ using a training set
$\mathcal{T}$ that only consists of sample points $\x$ close to $\x_o$.
This training set is created by a simple ABC rejection sampling algorithm.
More precisely: for a fixed distance function $d(\x,\x_o)$ (that could be based on summary statistics) and  tolerance level $\epsilon$,  
we construct a sample $\mathcal{T}$  according to 
Algorithm \ref{alg::naiveAbc}.
   \begin{algorithm}
  \caption{ \small Training set for CDE via Rejection ABC }\label{alg::naiveAbc}
  \algorithmicrequire \ {\small Tolerance level $\epsilon$, number of desired sample points $B$,  distance function $d$, sample $\x_0$}
  
  \algorithmicensure \ {\small Training set $\mathcal{T}$ which approximates the joint distribution of $(\theta,\X)$
  in a neighborhood of $\x_0$}
  \begin{algorithmic}[1]
  \State Let $\mathcal{T}=\{\}$
     \While{$|\mathcal{T}|<B$}
     \State Sample $\theta \sim f(\theta)$
\State  Sample $\X \sim f(\x|\theta)$
     \State  If $d(\x,\x_o)< \epsilon,$ let $\mathcal{T} \longleftarrow \mathcal{T} \cup \{(\theta,\x)\}$
     \EndWhile
     \State \textbf{return} $\mathcal{T}$  
  \end{algorithmic}
  \end{algorithm}
To this new training set $\mathcal{T}$, we then apply our conditional density estimator, and finally evaluate the estimate at $\x=\x_o$. 
This procedure can be regarded as an ABC post-processing technique \citep{marin2012approximate}:
 the  first  (ABC) approximation to the posterior
is obtained via the sample $\theta_1,\ldots,\theta_B$,
which can be seen as a sample from $f(\theta|d(\X,\x_o)<\epsilon)$.  That is, the standard ABC rejection sampler is implicitly performing conditional density estimation using an i.i.d.  sample
from the joint distribution of the data and the parameter.  We take the results of the ABC sampler and 
estimate  
the conditional density exactly at the point $\x=\x_o$ 
using other forms of conditional density estimation. 
If done correctly, the idea is that we can improve upon the original ABC approximation {\em even without}, as is currently the norm, simulating new data or decreasing the tolerance level $\epsilon$.

\begin{remark}  \label{remark::summary} For simplicity, we focus on standard ABC rejection sampling, but one can use other ABC methods, such as sequential ABC \citep{sisson2007sequential} or
population Monte Carlo ABC \citep{beaumont2009adaptive}, to construct 
$\mathcal{T}$. The data $\x$ can either be 
the original data vector, or a vector of
summary statistics. We revisit the issue of summary statistic selection in Section \ref{sec::summary}. 
\end{remark}

Next, we review FlexCode \citep{IzbickiLeeFlexCode}, which we currently use as a general-purpose methodology for estimating $f(\theta|\x)$. However, many aspects of the paper  (such as the novel approach to method selection 
without knowledge of the true posterior)
 hold for other CDE and ABC methods as well. 
In Section \ref{sec::beyond_ABC}, 
for example, we use our surrogate loss to choose the tuning parameters of a  nearest-neighbors kernel density estimator (Equation \ref{eq:kernel-nn-cde}), which includes ABC as a special case. 

\vspace{2mm}
\textbf{FlexCode as a ``Plug-In'' CDE Method.} For simplicity, 
 assume that we are interested in estimating the posterior distribution of a single parameter
$\theta \in \Re$, even if there are several parameters in the problem.\footnote{Most inference problems can be expressed as the computation of 
unidimensional quantities.
Say one is interested in estimating $m$ functions of
parameters of the model $\boldsymbol{\theta}$; $g_1,\ldots,g_m$.
One can then
(i) use ABC to obtain a single simulation set  $\mathcal{T}= \{(\boldsymbol{\theta}_1,\X_1), \ldots, (\boldsymbol{\theta}_B,\X_B)\}$,  (ii) 
for each function $g_i$, compute $\mathcal{T}^{g_i}= \{(g_i(\boldsymbol{\theta}_1),\X_1), \ldots, (g_i(\boldsymbol{\theta}_B),\X_B)\}$, and then (iii)
fit a (univariate) conditional density estimator to $\mathcal{T}^{g_i}$ to estimate $f(g_i(\boldsymbol{\theta})|\x_o)$. Note that (ii) is typically fast and (iii) can be performed in parallel; hence, the posterior distributions of all quantities of interest
can be estimated with essentially no additional computational cost.
}
Similar ideas can be used if one is interested in estimating the (joint) posterior distribution  for more than one parameter (see \citealt{IzbickiLeeFlexCode} for more details on how FlexCode can be adapted to those settings).
 In the context of ABC, $\x$ typically represents a set of statistics computed from the original data; recall Remark \ref{remark::summary}. We start by 
specifying an orthonormal basis $(\phi_i)_{i \in \mathbb{N}}$ in $\Re$.
This basis will be used to model the density $f(\theta|\x)$ {\em as a function of $\theta$}. 
Note that there is a wide range of (orthogonal) bases one can choose from to capture any challenging shape of the density function of interest 
\citep{mallat1999wavelet}.
For instance, a natural choice for reasonably smooth functions $f(\theta|\x)$ is the Fourier basis:
\[
\phi_1(\theta)=1;\hspace{9mm} \phi_{2i+1}(\theta)=\sqrt{2}\sin{\left(2\pi i\theta \right)}, \ i\in \mathbb{N}; \hspace{9mm} \phi_{2i}(\theta)=\sqrt{2}\cos{\left(2\pi i\theta \right)}, \ i\in \mathbb{N}
\]

The key idea of FlexCode is to notice that, if $\int  f^2(\theta|\x)d\theta < \infty$ for every $\x \in \mathcal{X}$,
then it is possible to expand $f(\theta|\x)$ as
$f(\theta|\x)=\sum_{i \in \mathbb{N}}\beta_{i }(\x)\phi_i(\theta),
$ where the expansion coefficients are given by
\begin{align}
\label{eq::beta}
\beta_{i }(\x) =  \E\left[\phi_i(\theta)|\x\right].
\end{align}
That is, each  $\beta_{i }(\x)$ is a {\em regression function}. 
The FlexCode estimator is defined as
$\widehat{f}(\theta|\x)=\sum_{i=1}^I\widehat{\beta}_{i}(\x)\phi_i(\theta),
$
where  $\widehat{\beta}_{i }(\x)$ are regression estimates. The cutoff $I$ in the series expansion is a tuning parameter that controls the bias/variance tradeoff in the final density estimate, and which we choose via data splitting (Section \ref{sec::comparing}). 

With FlexCode, the problem of high-dimensional conditional density estimation boils down to choosing appropriate methods for estimating the  regression functions $\E \left[\phi_i(\theta)|\x\right]$.  The key advantage of FlexCode is that it offers more flexible CDE methods: By taking advantage of existing regression methods, which can be ``plugged in'' into the CDE estimator, we can adapt to the intrinsic structure of high-dimensional data (e.g., manifolds, irrelevant covariates, and different relationships between $\x$ and the response $\theta$), as well as handle different data types (e.g., mixed data and functional data) and massive data sets (by using, e.g., xgboost \citep{chen2016xgboost}). See \citet{IzbickiLeeFlexCode} and the upcoming LSST-DESC photo-z DC1 paper for examples.
  An implementation of 
FlexCode  that allows for wavelet bases can be found at \url{https://github.com/rizbicki/FlexCoDE} (R; \citealt{R}) and \url{https://github.com/tpospisi/flexcode} (Python).

\subsection{Method Selection: Comparing Different Estimators of the Posterior}
\label{sec::comparing}

{\noindent \bf  Definition of a Surrogate Loss.}  Ultimately, we need to be able to decide which
approach is best for approximating $f(\theta|\x_o)$ without knowledge of the true posterior.
Ideally we would like to find an estimator
$\widehat{f}(\theta|\x_o)$ such that the   integrated squared-error (ISE)  loss
\begin{align}
\label{eq::trueLoss}
L_{\x_o}(\widehat{f},f)= \int (\widehat{f}(\theta|\x_o)-f(\theta|\x_o))^2  d\theta
\end{align}
is small. Unfortunately, one cannot compute $L_{\x_o}$ without knowing the true $f(\theta|\x_o)$,  which is why method selection is so hard in practice.
To overcome this issue, we propose the \emph{surrogate} loss function
\begin{align}
\label{eq::surrogateLoss}L_{\x_o}^\epsilon(\widehat{f},f)=\int \! \int (\widehat{f}(\theta|\x)-f(\theta|\x))^2 \frac{f(\x) \I(d(\x,\x_o)< \epsilon)}{\P(d(\X,\x_o)<\epsilon)} d\theta d\x, 
\end{align}
which enforces a close fit in an $\epsilon$-neighborhood of $\x_o$.
Here, the denominator $\P(d(\X,\x_o)<\epsilon)$
is simply a constant that makes $\frac{f(\x) \I(d(\x,\x_o)< \epsilon)}{\P(d(\X,\x_o)<\epsilon)}$
a proper density in $\x$. 

The advantage with the above definition is that we can directly {\em estimate} $L_{\x_o}^\epsilon(\widehat{f},f)$ from the ABC posterior sample. Indeed, it holds that $L_{\x_o}^\epsilon(\widehat{f},f)$ can be written as
\begin{align}
\label{eq::LHat}
&\int \! \int \widehat{f}^2(\theta|\x) \frac{f(\x) \I(d(\x,\x_o)< \epsilon)}{\P(d(\X,\x_o)<\epsilon)}  d\theta d\x -2 
\int \! \int  \widehat{f}(\theta|\x)f(\theta|\x) \frac{f(\x) \I(d(\x,\x_o)< \epsilon)}{\P(d(\X,\x_o)<\epsilon)} d\theta d\x
+ K_f \notag \\
&  =\E_{\X'}\left[\int \widehat{f}^2(\theta|\X') d\theta \right]-2 \E_{(\theta',\X')}\left[\widehat{f}(\theta'|\X')\right]+K_f,
\end{align}
where $(\theta',\X')$ is a random vector with distribution induced by a sample generated according to the ABC rejection procedure in Algorithm \ref{alg::naiveAbc};
 and $K_f$ is a constant that does not depend on the estimator $\widehat{f}(\theta|\x_o)$.
It follows that, given an independent validation or test sample of size $B'$ of the ABC algorithm,  
$(\theta'_1,\X'_1), \ldots, (\theta'_B,\X'_B)$,
we can estimate $L_{\x_o}^\epsilon(\widehat{f},f)$ (up to the constant $K_f$) via
\begin{align}
\label{eq::lossEstimate}
\widehat{L}_{\x_o}^\epsilon(\widehat{f},f)=\frac{1}{B'}\sum_{k=1}^{B'} \int \widehat{f}^2(\theta|\x'_k) d\theta  -2 \frac{1}{B'}\sum_{k=1}^{B'} \widehat{f}(\theta'_k|\x'_k)
 \end{align}
 When given a set of estimators $\mathcal{F}=\{\widehat{f}_1,\ldots,\widehat{f}_m\}$, 
we select
the method with the smallest estimated surrogate loss,
$$\widehat{f}^* := \arg \min_{\widehat{f} \in \mathcal{F}} \widehat{L}_{\x_o}^\epsilon(\widehat{f},f)$$

\begin{ex}[Model selection based on CDE surrogate loss versus regression MSE loss]
 Suppose we wish to estimate the posterior distribution of the mean of a Gaussian distribution with variance one. The left plot of Figure \ref{fig::CDEsReg} shows the performance of 
a 
 nearest-neighbors kernel density estimator
(Equation~\ref{eq:kernel-nn-cde}) with the kernel bandwidth $h$ and the number of nearest neighbors $k$ chosen via (i) the estimated surrogate loss of Equation \ref{eq::lossEstimate}
versus (ii) a standard regression mean-squared-error loss.\footnote{The data are simulated using the same Gaussian model as in Section~\ref{sec::beyond_ABC}, but  with $n=10$,
 $\bar{x}=0.5$ and at an acceptance ratio
 equal to 1 (that is, $\epsilon \rightarrow \infty$) and the number of simulations, B, varying.} The proposed surrogate loss clearly leads to better estimates of the posterior $f(\theta|\x_o)$ with smaller true loss (Equation \ref{eq::trueLoss}).
 Indeed, as the right plot shows, if one chooses tuning parameters via the standard regression mean-squared-error loss, the estimates end up being very far from the true distribution.
\end{ex}

  \begin{figure}[H]
  	\centering
  	\includegraphics[page=1,scale=0.29]{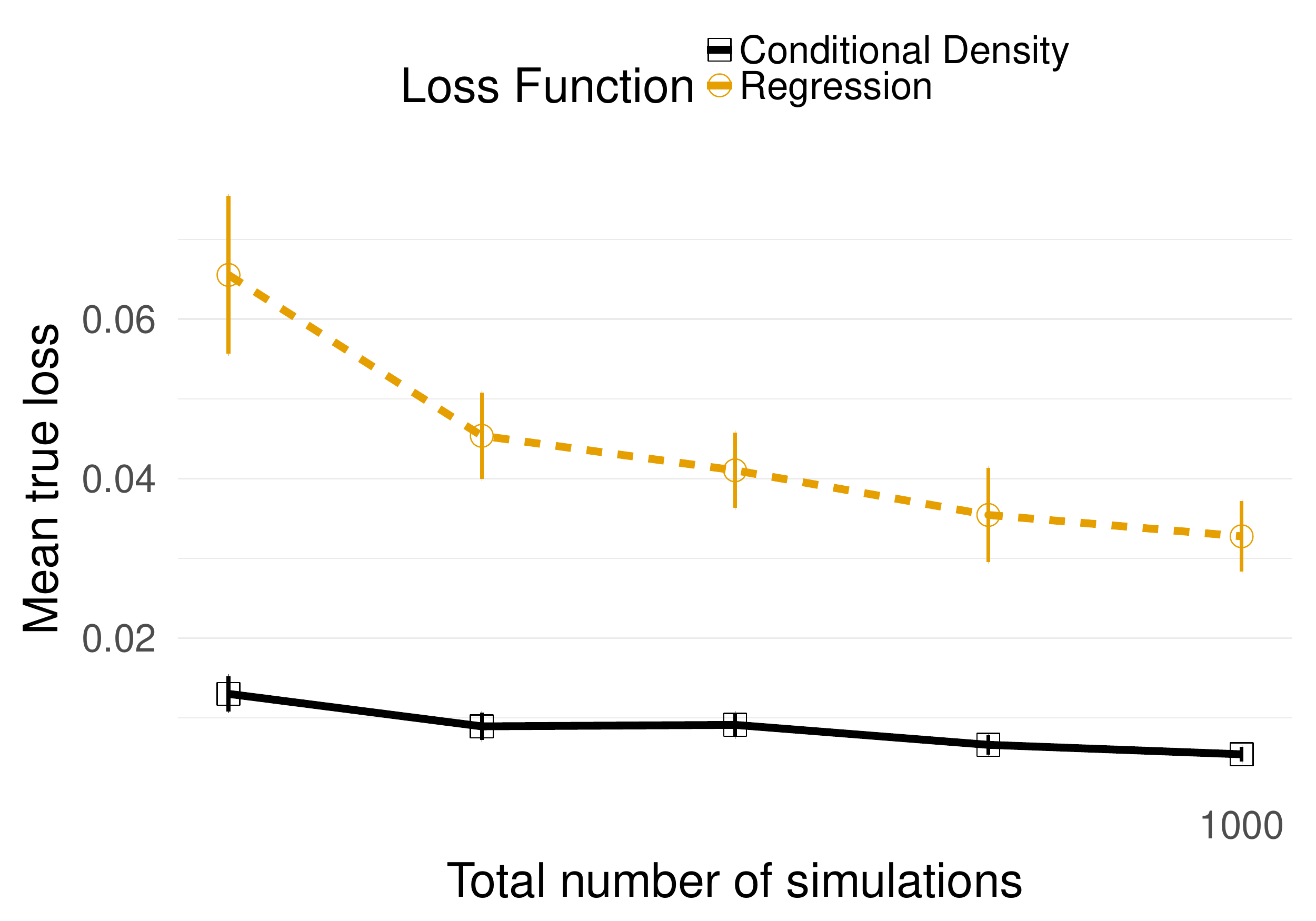}
  	  	\includegraphics[page=1,scale=0.29]{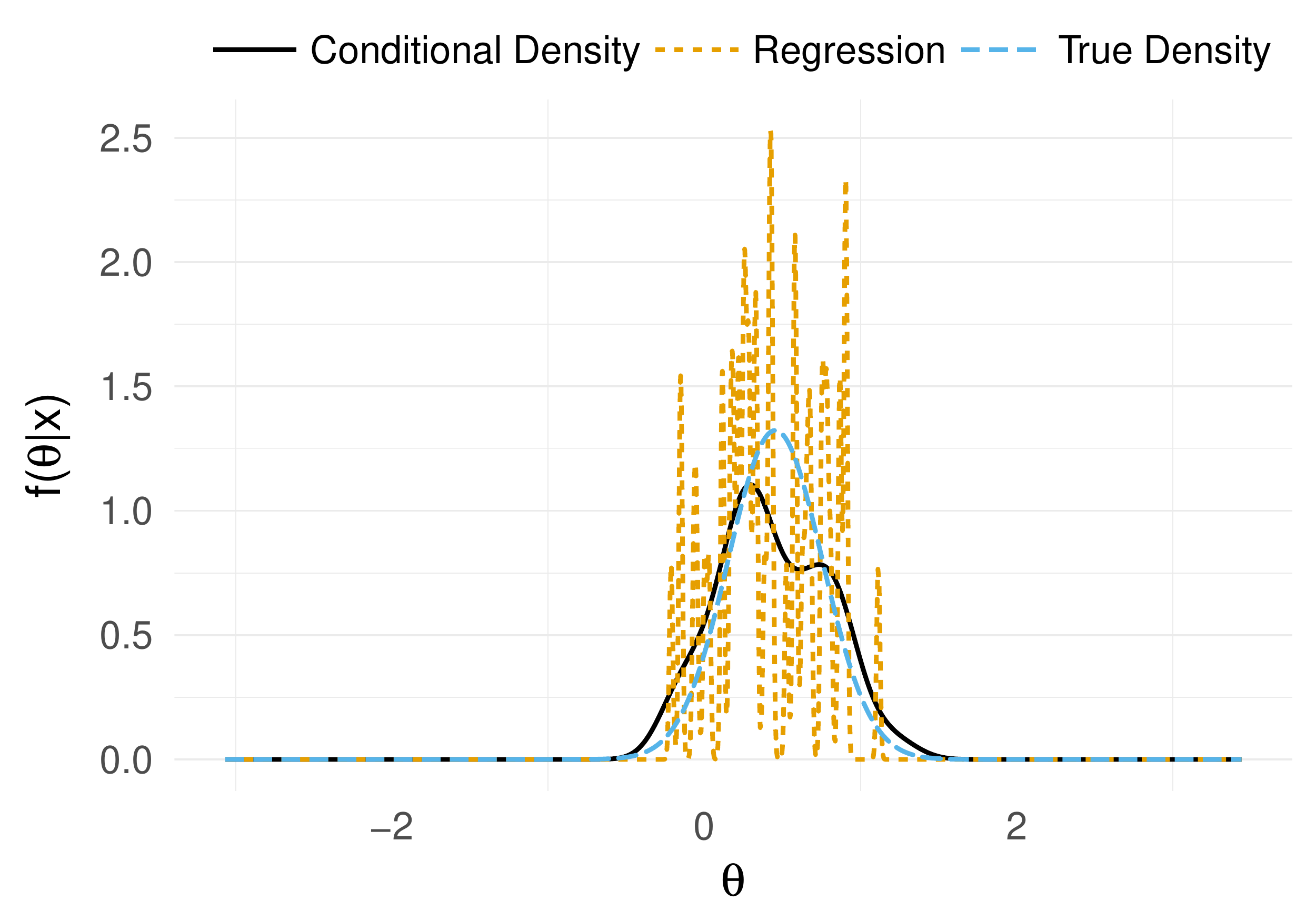}
  	\vspace{-3mm}
  	\caption{\footnotesize    {\em Left:} Performance of 
  	a nearest-neighbors kernel density estimator  with
  tuning parameters chosen via the surrogate loss of Equation \ref{eq::lossEstimate} (continuous line)
  	and via standard regression MSE loss (dashed line). {\em Right:} Estimated posterior distributions
  	tuned according to both criteria after 1000 simulations. The surrogate loss of Equation \ref{eq::lossEstimate} clearly leads to a better approximation.} 
  	\label{fig::CDEsReg}
  \end{figure}

{\noindent \bf  Properties of the Surrogate Loss.} 
Next we investigate the conditions under which the estimated surrogate loss  is close to the true loss; the proofs can be found in Appendix. The following theorem states that, if $(\widehat{f}(\theta|\x)-f(\theta|\x))^2$ is a smooth function of $\x$, then the (exact) surrogate loss $L_{\x_o}^\epsilon$ is close to $L_{\x_o}$ for small values of $\epsilon$.  

\begin{thm}
\label{Eq::LossBias}
Assume that, for every $\theta \in \Theta$, $g_\theta(\x):= (\widehat{f}(\theta|\x)-f(\theta|\x))^2$
satisfies the  H\"{o}lder condition of order $\beta$ with a constant $K_\theta$\footnote{That is, there exists a constant $K_\theta$ such that for every $\x,\vec{y} \in \Re^d$  $|g_\theta(\x)-g_\theta(\vec{y})|\leq K_\theta (d(\x,\vec{y}))^\beta$. } such that $K_H:=\int K_\theta d\theta < \infty$. Then
$|L_{\x_o}^\epsilon(\widehat{f},f)-L_{\x_o}(\widehat{f},f)| \leq K_H \epsilon^\beta = O(\epsilon^\beta)$
\end{thm}

The next theorem shows that the estimator $\widehat{L}_{\x_o}^\epsilon$ in Equation \ref{eq::lossEstimate}
does indeed converge to the true loss $L_{\x_o}(\widehat{f},f)$. 

\begin{thm}
Let $K_f$ be as in Equation \ref{eq::LHat}. Under the assumptions of Theorem \ref{Eq::LossBias},
$|\widehat{L}_{\x_o}^\epsilon(\widehat{f},f)+K_f-L_{\x_o}(\widehat{f},f)| = O(\epsilon^\beta)+O_P(1/\sqrt{B'})$
\end{thm}

%

Under some additional conditions, it is also possible to guarantee that not only 
the estimated surrogate loss is close to the true loss, but that the result holds 
uniformly for a finite class of estimators of the posterior distribution. This is formally stated in the following theorem.

\begin{thm}
\label{thm::estimateToTrueUnion}
Let $\mathcal{F}=\{\widehat{f}_1,\ldots,\widehat{f}_m\}$ be a   set of estimators of $f(\theta|\x_o)$.
Assume that there exists $M$ such that $|\widehat{f}_i(\theta|\x)|\leq M$ for every $\x$, $\theta$, and $i=1,\ldots,m$.\footnote{Such
assumptions hold if the $\widehat{f}_i$'s are obtained
via FlexCode with bounded basis functions (e.g., Fourier basis)
or a kernel density estimator on the ABC samples.}
Moreover, assume that
for every $\theta \in \Theta$, $g_{i,\theta}(\x):= (\widehat{f}_i(\theta|\x)-f(\theta|\x))^2$
satisfies the  H\"{o}lder condition of order $\beta$ with constants $K_\theta$
such that $ K_H:=\int K_\theta d\theta < \infty$. Then, for every $\nu>0$,
 $$\P\left(\max_{\widehat{f} \in \mathcal{F}} |\widehat{L}_{\x_o}^\epsilon(\widehat{f},f)+K_f-L_{\x_o}(\widehat{f},f)| \geq K_\epsilon \epsilon^\beta+	\nu\right) \leq  2m e^{-\frac{B'\nu^2}{2(M^2+2M)^2}}.$$
\end{thm}

The next corollary shows that the procedure
we propose in this section, with high probability, picks an estimate of the posterior
density that has a true loss that is close to the true  loss of the best method in $\mathcal{F}$. 
\begin{Cor}
Let
$\widehat{f}^* := \arg \min_{\widehat{f} \in \mathcal{F}} \widehat{L}_{\x_o}^\epsilon(\widehat{f},f)$
be the best estimator in $\mathcal{F}$ according to the estimated surrogate loss, and
 let $f^*=\arg \min_{\widehat{f} \in \mathcal{F}} L_{\x_o}(\widehat{f},f)$
be the best estimator in $\mathcal{F}$ according to the true loss.
Then, under the assumptions from Theorem \ref{thm::estimateToTrueUnion}, with probability at least $1-2m e^{-\frac{B'\nu^2}{2(M^2+2M)^2}}$,
$ L_{\x_o}(\widehat{f}^*,f)\leq  L_{\x_o}(f^*,f) +2 (K_H\epsilon^\beta+	\nu).$
\end{Cor}

\subsection{Summary Statistics Selection}
\label{sec::summary}

In a typical ABC setting, there are a large number of available summary statistics. 
 Standard ABC fails if all of them are used simultaneously, especially if some statistics carry little information 
about the parameters of the model \citep{blum2010approximate}.

One can use 
ABC-CDE 
as a way of either (i) directly estimating
$f(\theta|\x_o)$ when there are a large number of summary statistics,\footnote{the dimension reduction is then implicit in the choice of (high-dimensional) regression method} or (ii)
 assigning an importance measure to each summary statistic 
  to guide variable selection in ABC and related procedures.

There are two versions of 
ABC-CDE  that are particularly useful for variable selection:
FlexCode-SAM\footnote{FlexCode with the coefficients from Equation~\ref{eq::beta} 
estimated via Sparse Additive Models \citep{Ravikumar:EtAl:2009}} and 
FlexCode-RF\footnote{FlexCode with the coefficients from Equation~\ref{eq::beta} 
estimated via Random Forests}. 
\citet{IzbickiLeeFlexCode}
show that both estimators automatically adapt to the number
of relevant covariates, i.e., the number of 
covariates that influence the distribution of the response. In the context of ABC, this means
that these methods are able to automatically detect which summary statistics are relevant in estimating the posterior distribution
of $\theta$. 
Corollary 1
	from \citet{IzbickiLeeFlexCode} implies that, {\em if} indeed 
 only $m$ out of all $d$ summary statistics influence the distribution of $\theta$, then the rate of convergence of these methods
 is $O\left(n^{-2\beta/(2\beta+m\frac{2\beta+1}{2\alpha}+1)}\right)$
 instead of $O\left(n^{-2\beta/(2\beta+d\frac{2\beta+1}{2\alpha}+1)}\right)$,
 where $\alpha$ and $\beta$ are numbers associated to the smoothness of $f(\theta|\x)$.  The former rate implies a much faster convergence: if $m \ll d$, it is essentially
 the rate one would obtain if one knew which were the relevant statistics.
 In such a setting, there is no need to explicitly perform summary statistic selection prior to estimating the posterior; 
 FlexCode-SAM or 
 FlexCode-RF automatically remove
 irrelevant covariates. 

More generally, one can use FlexCode to compute an importance measure for summary statistics (to be used in other procedures than FlexCode). It turns out that one can infer the relevance of the $j$:th summary statistic {\em in posterior estimation} from its relevance in estimating the $I$ first {\em regression} functions in FlexCode --- even if we do not use FlexCode for estimating the posterior. 
	More precisely, assume that  $\x=(x_1,\ldots, x_j, \ldots, x_d)$ is a vector of summary statistics, and let $\x'_j=(x_1,\ldots, x_{j-1}, x'_j, x_{j+1},\ldots, x_d)$.
	Define the relevance of variable $j$ to the posterior distribution $f(\theta|\x)$ as 
	$$r_j:=\int \! \int \! \int (f(\theta|\x)-f(\theta|\x'_j))^2 d\x dx'_j d\theta,$$
	and its relevance to the regression $\beta_i(\x)$ in Equation \ref{eq::beta} as
		$$r_{i,j}:=\int \! \int \left(\beta_i(\x)-\beta_i(\x'_j)\right)^2 d\x dx'_j.$$ Under some smoothness assumptions with respect to $\theta$, the two metrics are related.


\begin{Assumption}[Smoothness in $\theta$ direction]
\label{label:assumpSmooth} 
 \label{assump-sobolevZ} $\forall \vec{x} \! \in \! \mathcal{X}$, we assume that
$f(\theta|\x) \! \in \! W_{\phi}(s_\vec{x},c_\vec{x}),$ the
Sobolev space of order $s$ and radius $c$,\footnote{For every $ s>\frac{1}{2}$ and $0<c<\infty $, $W_{\phi}(s,c):=\{f \!= \!\sum_{i\geq 1} \theta_i \phi_i \!:  \! \sum_{i\geq 1} a_i^2 \theta^2_i \leq c^2 \}$, where $a_i \! \sim \! (\pi i)^s$. Notice that for the Fourier basis
 $(\phi_i)_i$, this is the standard definition of the Sobolev space of order $s$ and radius $c$; it is  the
 space of functions that have their $s$-th weak derivative bounded by $c^2$ and integrable in $L^2$.
} 
 where $f(\theta|\x)$ is viewed as a function of $\theta$, and $s_\vec{x}$ and $c_\vec{x}$ are such that 
$\inf_\vec{x} s_\vec{x}\overset{\mbox{\tiny{def}}}{=}\beta>\frac{1}{2}$ and 
$\int c^2_\vec{x} d\x <\infty$.
\end{Assumption}

\begin{prop} Under Assumption \ref{label:assumpSmooth},
$ r_j =  \sum_{i=1}^I r_{i,j} + O\left(I^{-2\beta}\right)$ \label{prop::stats_relevance}
\end{prop}

Now let $u_{i,j}$ denote a measure of importance of the $j$:th summary statistic in estimating 
regression $i$ (Equation~\ref{eq::beta}). For instance, for 
FlexCode-RF, $u_{i,j}$ may represent the mean decrease in the Mean Squared Error \citep{Hast:Tibs:Frie:2001};
for FlexCode-SAM, $u_{i,j}$ may be value of the indicator function for the $j$:th summary statistic 
 when estimating $\beta_i(\x)$. Motivated by Proposition~\ref{prop::stats_relevance}, we define an \emph{importance measure} for the $j$:th summary statistic {\em in posterior estimation} according to 
\begin{align}
\label{eq::importance}
u_j := \frac{1}{I}\sum_{i=1}^I u_{i,j}.
\end{align}
We can use these values to select variables for estimating $f(\theta|\x_o)$
via other ABC methods. For example, one approach is to choose all summary statistics such that $u_j>t$, where the threshold value $t$ is defined by the user.
We will further explore this approach in Section \ref{sec::experiments}.

In summary, our procedure has two main advantages compared to current state-of-the-art approaches for selecting summary statistics in ABC:
(i) it chooses statistics
that lead to good estimates of the   \emph{entire} posterior distribution $f(\theta|\x_o)$ rather than surrogates,
such as, the regression or  posterior mean $\E[\theta|\x_o]$  \citep{aeschbacher2012novel,creel2016selection,faisal2016choosing}, 
and
(ii) it is typically faster than most other approaches; in particular, it is significantly faster than best subset selection which scales as $O(2^d)$, whereas, e.g., FlexCode-RF scales as $O(Id)$, and
FlexCode-SAM scales as $O(Id^3)$.


\section{Experiments}
\label{sec::experiments}

\subsection{Examples with Known Posteriors}
\label{sec::toy}

We start by analyzing examples with well-known and analytically computable posterior distributions:
\begin{itemize}
		\item[]\textbf{1. Mean of a Gaussian with known variance.} $X_1,\ldots,X_{20}|\mu \overset{iid}{\sim}  \operatorname{Normal}(\mu,1)$, $\mu \sim  \operatorname{Normal}(0,\sigma^2_0)$. We repeat the experiments for $\sigma_0$ in an equally spaced grid with ten values between 0.5 and 100.

		\item[]\textbf{2. Precision of a Gaussian with unknown precision.} 			$X_1,\ldots,X_{20}|(\mu,\tau) \overset{iid}{\sim}  \operatorname{Normal}(\mu,1/\tau)$, $(\mu,\tau) \sim \mbox{Normal-Gamma}(\mu_0,\nu_0,\alpha_0,\beta_0)$. We set $\mu_0=0$, $\nu_0=1$,
			and repeat the experiments  choosing $\alpha_0$ and $\beta_0$ 
			such that $\E[\tau]=1$
			and $\sqrt{\V[\tau]}$ is in an equally spaced grid with ten  values between 0.1 and 5.
			
\end{itemize}

In the Appendix we also investigate a third setting, ``Mean of a Gaussian with unknown precision'', with results similar to those shown here in the main manuscript.

 In all examples here, \emph{observed} data    $\x_o$ are drawn from a $ \operatorname{Normal}(0,1)$ distribution.
		We run each
		experiment 200 times, that is, with 200 different values of $\x_o$. 
The training set $\mathcal{T}$, which is used to build conditional density estimators,
		is constructed according to Algorithm \ref{alg::naiveAbc} with $B=10,000$ and
		a tolerance level $\epsilon$ that corresponds to an acceptance rate of 1\%.
					For the distance function    $d(\x,\x_o)$, we choose the Euclidean distance between
					minimal sufficient statistics normalized to have mean zero and variance 1;                                        these statistics are  $\bar{\x}$
                                        for scenario 1 and
                                        $(\bar{\x},s)$ for scenario 2.
We use a Fourier basis for all FlexCode experiments in the paper, but wavelets lead to similar results.

	We compare the following methods (see Section \ref{sec::beyond_ABC} and the appendix for a comparison between ABC, regression adjustment methods and ABC-CDE with a standard kernel density estimator):
	\begin{itemize}
	\item ABC: rejection ABC method  with the minimal sufficient statistics (that is, apply a kernel density estimator to the $\theta$
coordinate of $\mathcal{T}$, 
	with bandwidth chosen via cross-validation),
	\item FlexCode\_Raw-NN: FlexCode estimator with Nearest Neighbors regression,  
	\item FlexCode\_Raw-Series: FlexCode estimator with Spectral Series regression \citep{LeeIzbickiReg}, and
	\item FlexCode\_Raw-RF: FlexCode estimator with Random Forest regression.
	\end{itemize}
 The three FlexCode estimators (denoted by ``Raw'') are directly applied to the sorted values of the {\em original} covariates $X_{(1)},\ldots,X_{(20)}$. That is, we do not use minimal sufficient statistics or other summary statistics. 
	To assess the performance of each method, we compute the true loss $L_{\x_o}$
	(Equation~\ref{eq::trueLoss}) for each $\x_o$. In addition, we estimate the surrogate loss 	$L_{\x_o}^\epsilon$ according to Equation \ref{eq::lossEstimate}   
 using a new sample of size $B'=10,000$ from Algorithm \ref{alg::naiveAbc}.

  \subsubsection{CDE and Method Selection}
  \label{sec::estPostModelSelec}
  
 In this section, we investigate whether various ABC-CDE methods 
 improve upon standard ABC for the settings described above. We also evaluate the method selection approach in Section \ref{sec::comparing} by comparing decisions based on estimated surrogate losses to those made if one knew the true ISE losses. 
  

  Figure \ref{fig::toyExample1}, left, shows how well the methods actually estimate the posterior density for Settings 1-2. Panel (a) 
  and (e) list the proportion of times each method returns the best results (according to the true loss from Equation~\ref{eq::trueLoss}). Generally speaking, the larger the prior variance, the better
  ABC-CDE   methods 
  perform compared to ABC. In particular, while for small variances ABC tends to be better, for
  large prior variances,
  FlexCode with Nearest Neighbors regression tends to give the best results.
    FlexCode with expansion coefficients estimated via Spectral Series regression is also very competitive.
Panels (c) and (g) confirms these results; here we see the average true loss of each method along with standard errors.

%
%
%
    Figure \ref{fig::toyExample1}, right, summarizes the performance of our method selection algorithm.
      Panels (b) and (f) list the proportion of times the method chosen
      by the true loss (Equation~\ref{eq::trueLoss}) matches the method chosen via the estimated loss (Equation~\ref{eq::lossEstimate}) in all pairwise comparisons;
      that is, the plot tells us how often the method selection procedure proposed in Section \ref{sec::comparing} actually works. 
       We present two variations of the algorithm: 
      in the first version (see triangles), we include all the data; in the second version (see circles), we remove cases where the confidence interval for $\widehat{L}^\epsilon_{\x_o}(\widehat{f}_1,f)-\widehat{L}^\epsilon_{\x_o}(\widehat{f}_2,f)$
      contains zero (i.e., cases where we cannot tell whether $\widehat{f}_1$ or
    $\widehat{f}_2$ performs better). The baseline 
      shows what one would expect if the method selection algorithm was totally random. The plots indicate that we, in all settings, roughly arrive at the same conclusions
    with the estimated surrogate loss as we would if we knew the true loss.  
    
    For the sake of illustration, we have also added panels (d)
    and (h), 
    %
     which show a scatterplot of differences between true losses versus the differences between the estimated losses for ABC and FlexCode\_Raw-NN
     for the setting with $\sigma_0=0.5$ and $\sqrt{\V[\tau]}=0.1$,
     respectively. The fact that most samples are either in the  first or third quadrant further confirms that the estimated surrogate loss is in agreement with the true loss in terms of which method best estimates the posterior density.

       \begin{figure}[H]
    	\centering
    	\subfloat[]{  \includegraphics[page=1,scale=0.212]{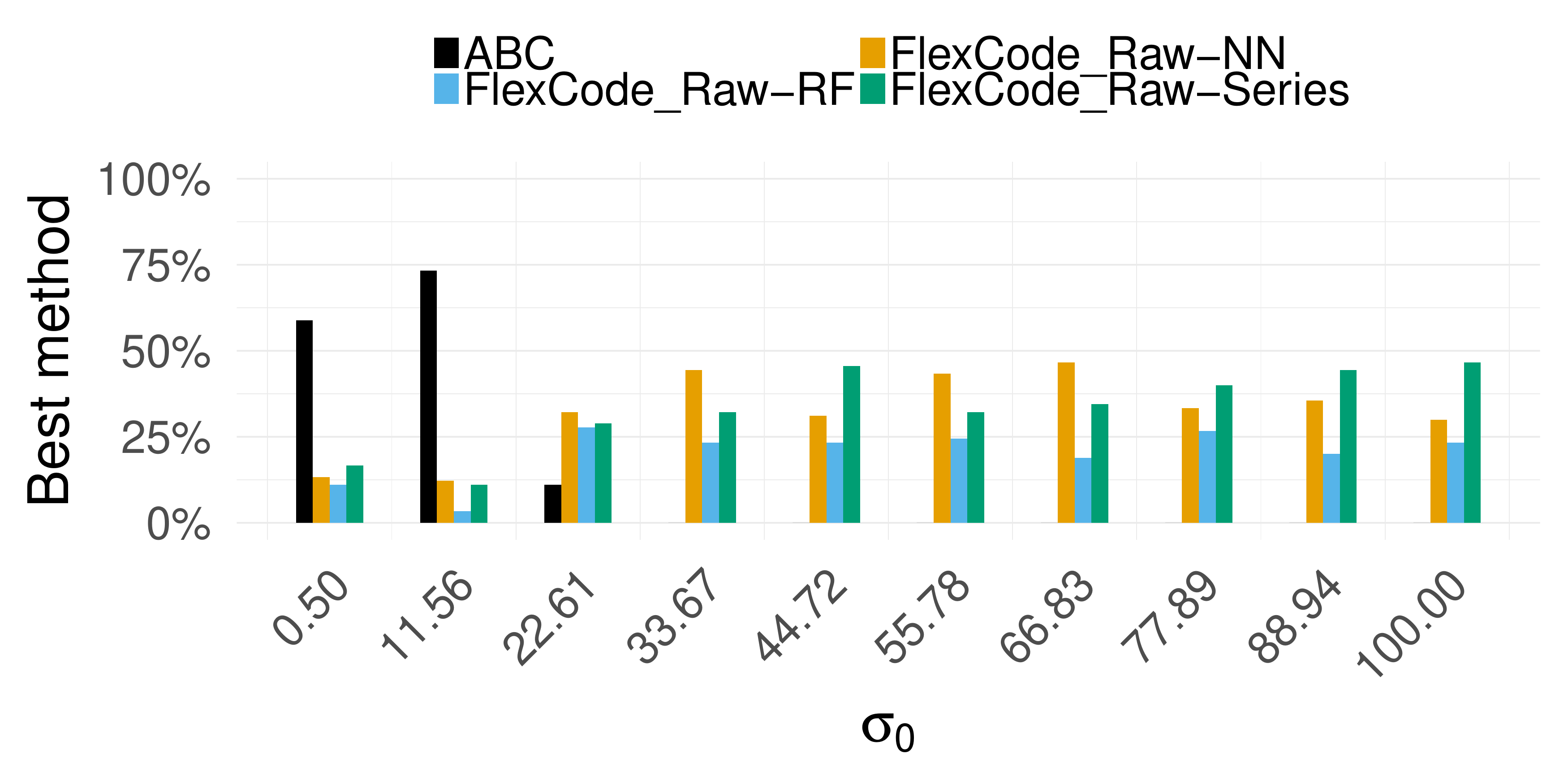}} 
    	\subfloat[]{  \includegraphics[page=1,scale=0.212]{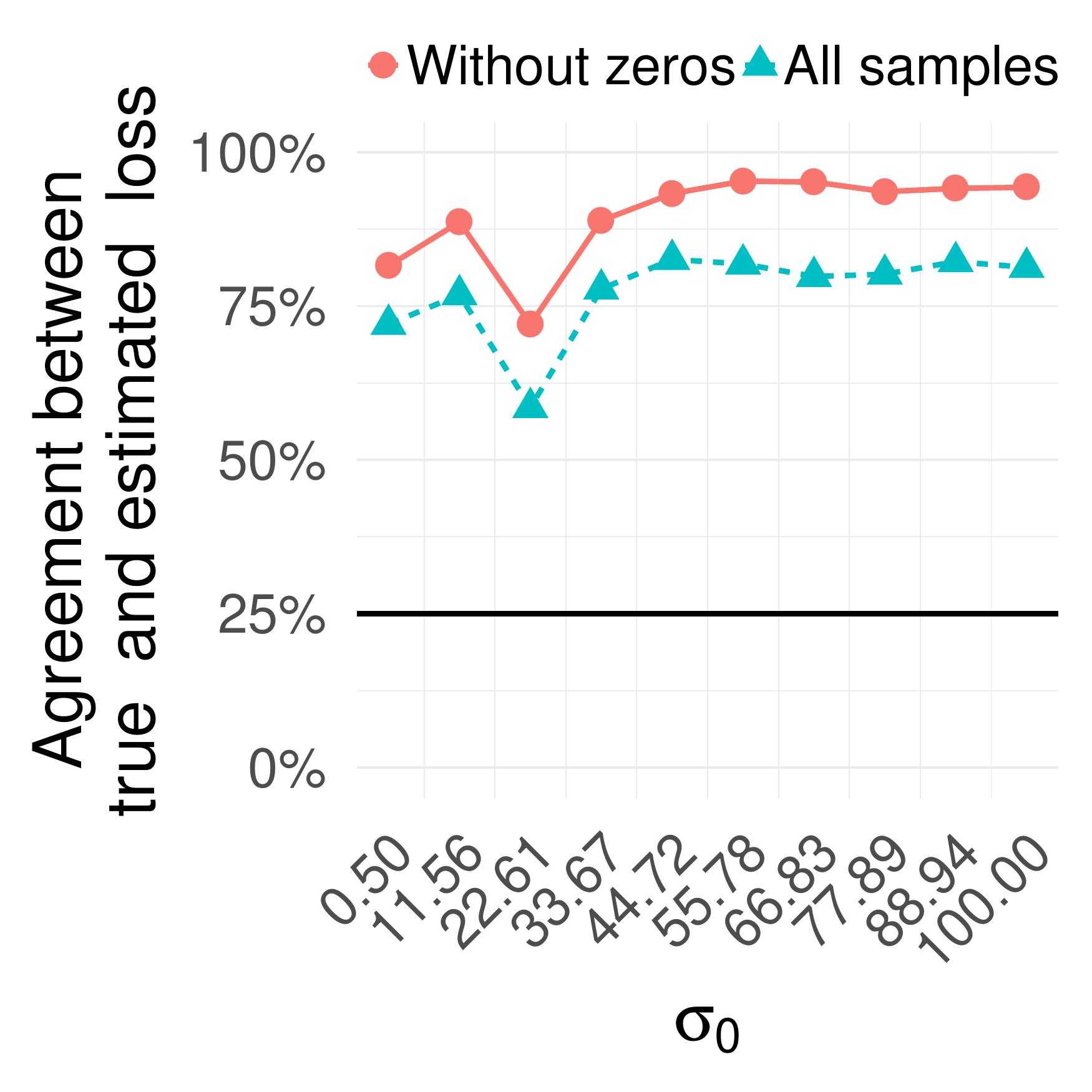}}  \\[-6mm]
    	  	\subfloat[]{  \includegraphics[page=2,scale=0.212]{figures/gauss_knownVar_reducedv2.pdf}} 
    	  	\subfloat[]{  \includegraphics[page=3,scale=0.212]{figures/gauss_knownVar_scatter_reducedv2.pdf}} \\[2mm]
    	  	    	\subfloat[]{  \includegraphics[page=1,scale=0.212]{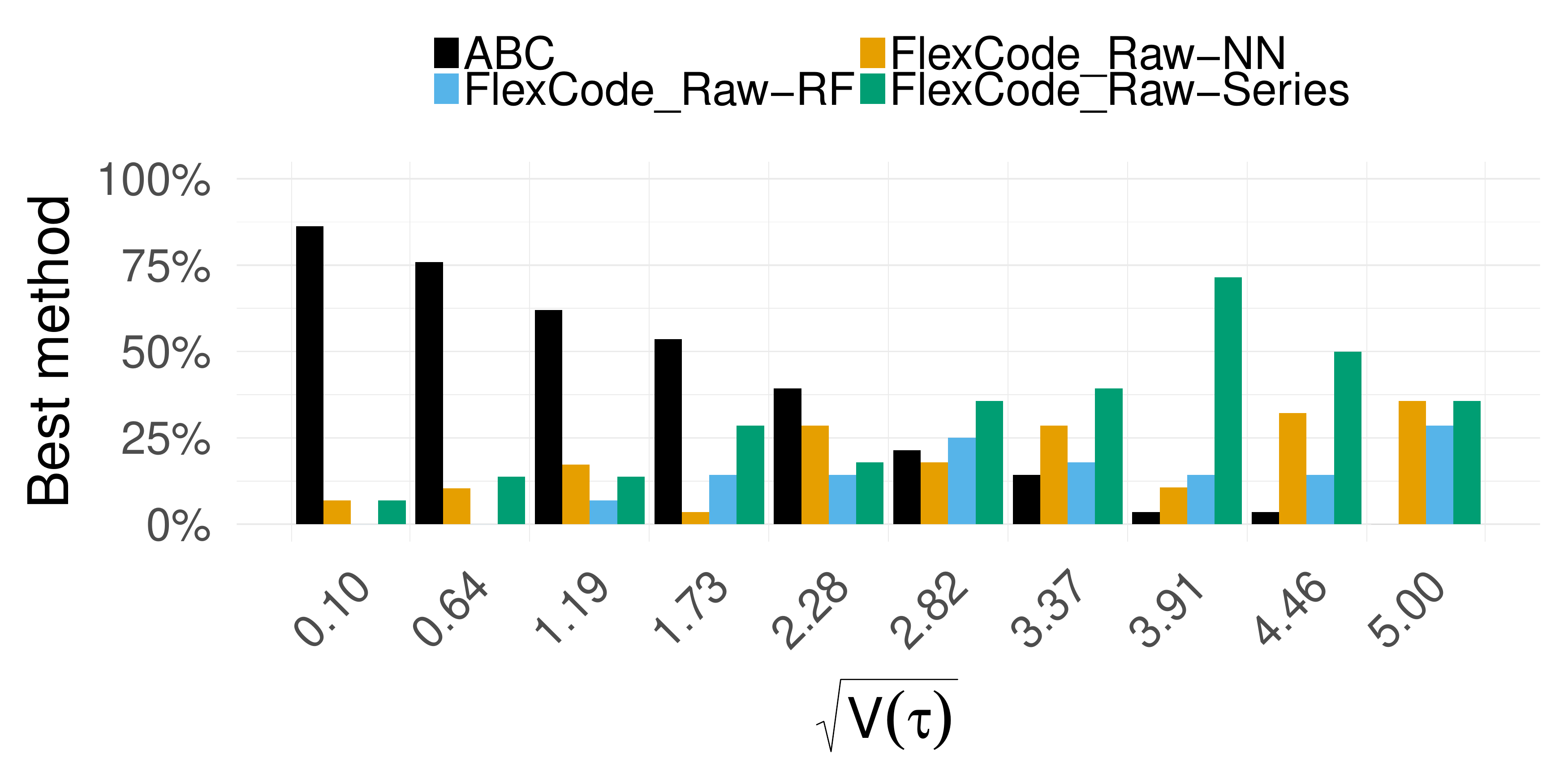}} 
    	  	    	\subfloat[]{  \includegraphics[page=1,scale=0.212]{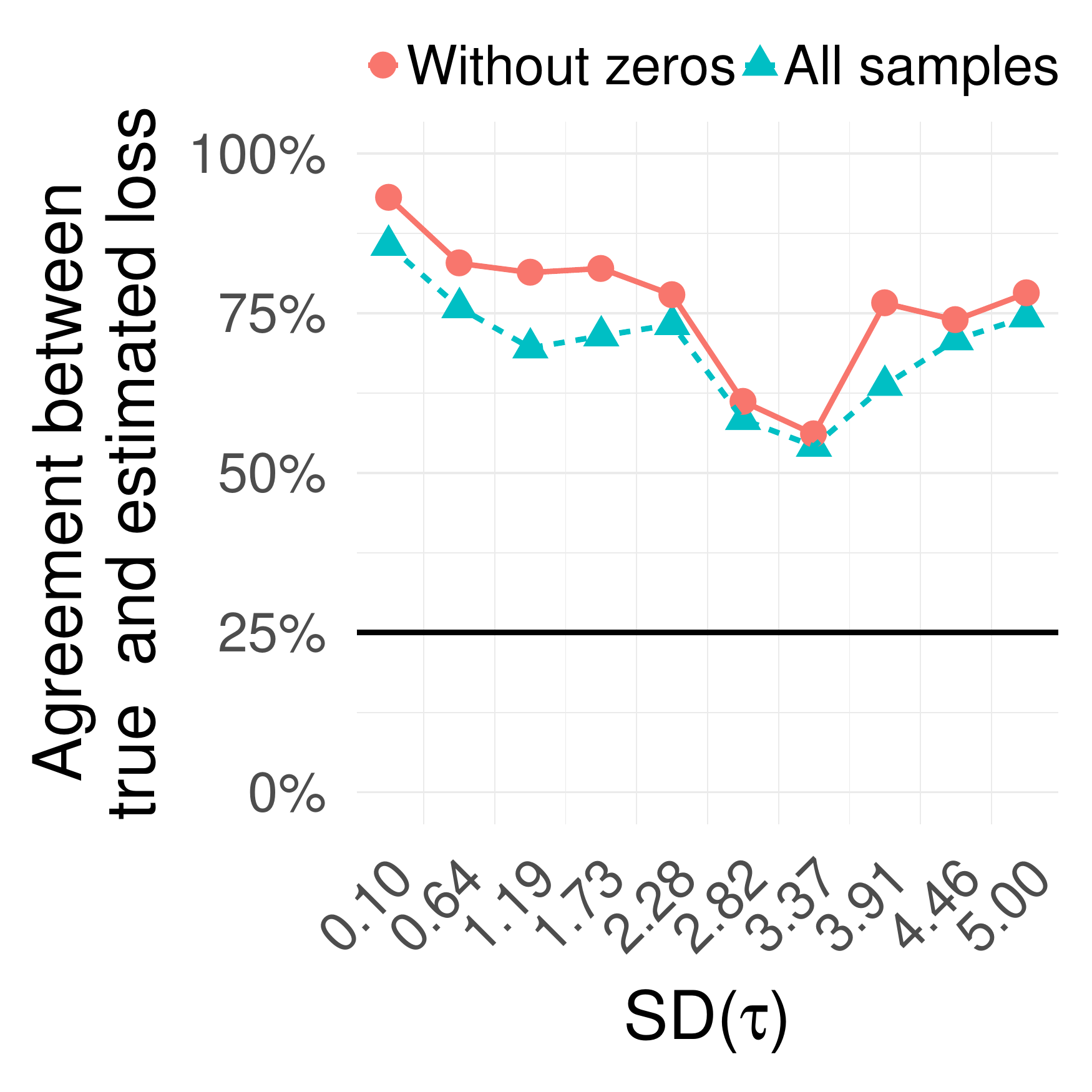}}  \\[-6mm]
    	  	    	\subfloat[]{  \includegraphics[page=2,scale=0.212]{figures/gauss_unknownVarSigma_reducedv2.pdf}} 
    	  	    	\subfloat[]{  \includegraphics[page=2,scale=0.212]{figures/gauss_unknownVarSigma_scatter_reducedv2.pdf}}  
    	  	 
    	\vspace{-3mm}
    	\caption{\footnotesize  Panels (a)-(d): CDE and method selection results for scenario 1 (mean of a Gaussian with known variance). {\em Left:}	 Panels (a) and (c) show that the
    	rejection ABC leads to better estimates of the posterior density $f(\theta|\x_o)$ when the prior variance $\sigma_0$ is small, but the NN and Series versions of FlexCode yield better estimates for moderate and large values of $\sigma_0$. {\em Right:} Panels (b) and (d) indicate that  by estimating the surrogate loss function one can tell from the data which method is better for the problem at hand. The horizontal line in panel (b) represents the behavior of a random selection.
    	Panels (e)-(h): CDE and method selection results for scenario 2 (precision of a Gaussian with unknown precision). 	
    	Conclusions are analogous.} 
    	\label{fig::toyExample1}
    \end{figure}

%
%

\begin{figure}[H]
      	\centering
      	\subfloat[]{  \includegraphics[page=2,scale=0.21]{figures/summary_setting1v2.pdf}} \\[-2mm]
      	\subfloat[]{  \includegraphics[page=1,scale=0.17]{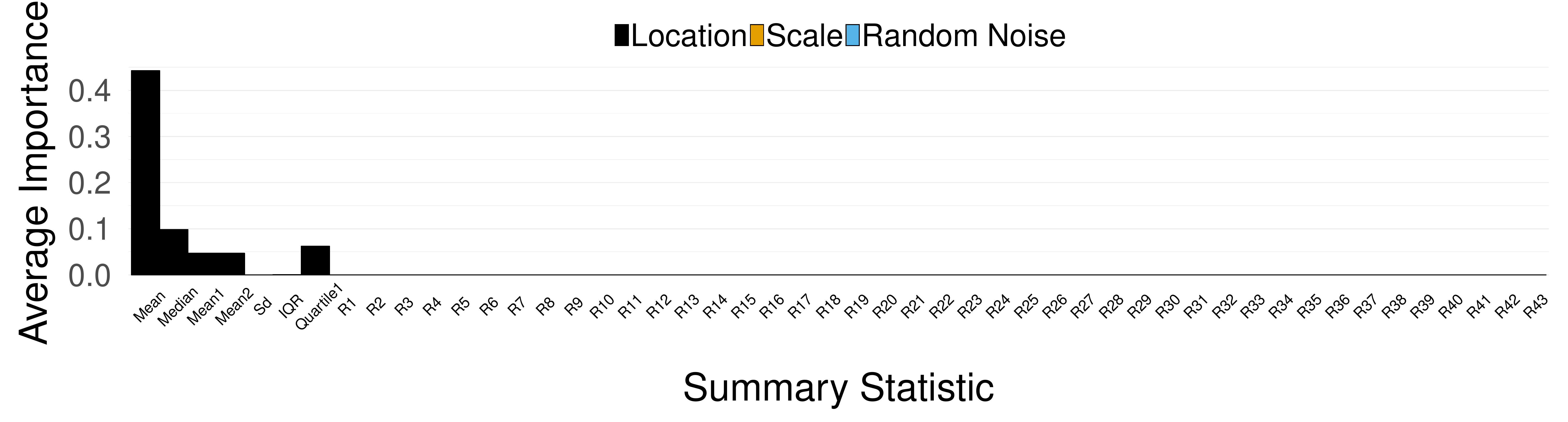}} \\
      	      	\subfloat[]{  \includegraphics[page=2,scale=0.21]{figures/summary_setting3v2.pdf}} \\[-2mm]
      	      	\subfloat[]{  \includegraphics[page=1,scale=0.17]{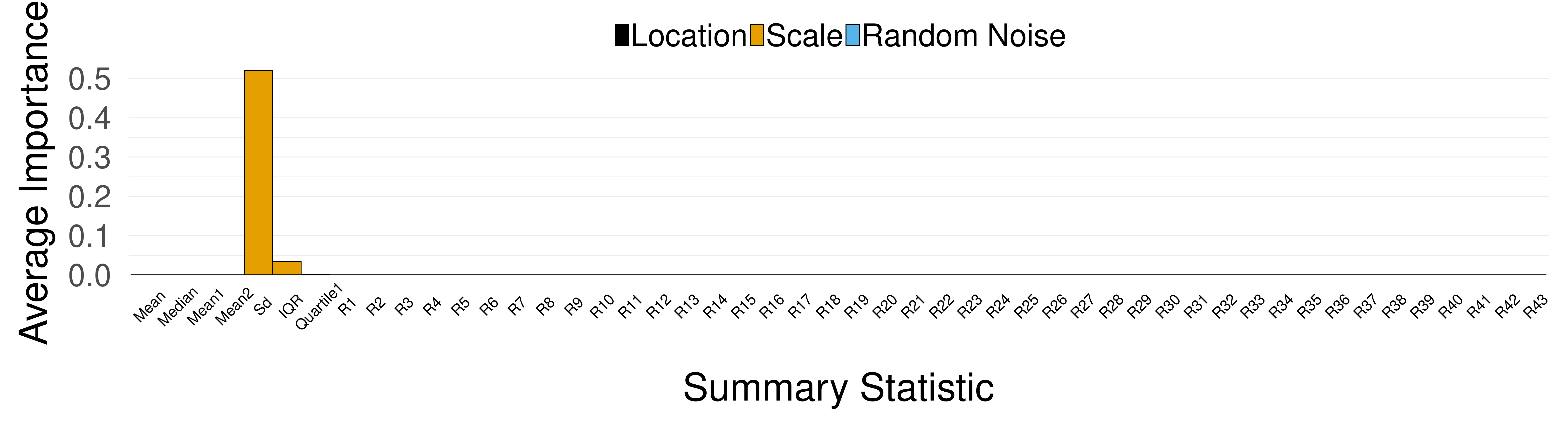}}  
      	\vspace{-3mm}
      	\caption{\footnotesize   
     Panels (a)-(b): Summary statistic selection for scenario 1 (mean of a Gaussian with known variance);
     Panels (c)-(d): Summary statistic selection  for scenario 2 (precision of a Gaussian with unknown precision). 	
      	Panels (a) and (c) show that
      	ABC  is highly sensitive to random noise (entries 8-51) with the estimates of the posteriors rapidly deteriorating with nuisance statistics.  Nuisance statistics do not affect the performance of FlexCode-RF much.
     Furthermore, we see from panel (b) that FlexCode-RF identifies the location statistics  (entries 1-5) as key variables for the first setting and assigns them a high average importance score.
     In the second setting,
     panel (d) indicates that we only need dispersion statistics (such as entry 5) to estimate the posteriors wells. } 
      	\label{fig::toyExampleSummary1}
      \end{figure}
   
       \subsubsection{Summary Statistic Selection}
    \label{sec::variableSelection}

    In this Section we investigate the performance
      of FlexCode-RF for summary statistics selection (Sec. \ref{sec::summary}).
      For this purpose,
     the following
 summary statistics were used:      \begin{enumerate}
     \item Mean: average of the data points; $\frac{1}{n}\sum_{i=1}^{n}X_i$ 
     \item Median: median of the data points; $\mbox{median}\{X_i\}_{i=1,\ldots,n}$
     \item Mean 1: average of the first half of the data points; $\frac{1}{n/2}\sum_{i=1}^{n/2}X_i$ 
     \item Mean 2: average of the second half of the data points;  $\frac{n/2+1}{n}\sum_{i=n/2+1}^{n}X_i$ 
     \item SD: standard deviation of the data points; $\sqrt{\frac{1}{n}\sum_{i=1}^{n}(X_i-\bar{\X})^2}$ 
     \item IQR: interquartile range of the data points; $\mbox{quantile}_{75\%}\{X_i\}_{i=1,\ldots,n}-\mbox{quantile}_{25\%}\{X_i\}_{i=1,\ldots,n}$ 
     \item Quartile 1: first quantile of the data points; $\mbox{quantile}_{25\%}\{X_i\}_{i=1,\ldots,n}$  
     \itemrange{43} Independent random variables  $\sim  \operatorname{Normal}(0,1)$, that is, random noise
     \end{enumerate}

Figure \ref{fig::toyExampleSummary1} summarizes the results of fitting FlexCode-RF and ABC to these summary statistics for the different scenarios.
      Panel (a) and (c) show the true loss as we increase the number of statistics. More precisely:
      the values at $x=1$ represent the true loss of ABC (left) and FlexCode-RF (right) when using only the mean (i.e., the first statistic);
        the points at $x=2$ indicate the true loss of the estimates using only the mean
  and the median (i.e., the first and second statistics) and so on. We note that
  %
   FlexCode-RF is robust to irrelevant summary statistics: the method virtually behaves as if they were not present.
      This is in sharp contrast with standard ABC, whose performance deteriorates quickly with added noise or nuisance statistics.

     Furthermore, panels (b) and (d) show the average importance of each statistic, defined according to Equation \ref{eq::importance},
      where $u_{i,j}$ is the
       mean decrease in the Gini index.
        These plots reveal that FlexCode-RF typically assigns a high score to sufficient summary statistics or to statistics
      that are highly correlated to sufficient statistics. For instance, in panel (b)
       (estimation of the mean of the distribution),
       measures of location are assigned a higher importance score, whereas  
       measures of dispersion are assigned a higher score in  panel (d) (estimation of the precision of the distribution). In all examples, FlexCode-RF assigns zero importance to random noise statistics. We conclude that our method for summary statistic selection indeed identifies relevant statistics for estimating the posterior $f(\theta|\x_o)$ well. 
\section{Approximate Bayesian Computation and Beyond}
\label{sec::beyond_ABC}

In this section, we show how one can use our surrogate loss to choose the tuning parameters in standard ABC with a nearest neighbors kernel smoother.

\subsection{ABC with Fewer Simulations} \label{sec::ABC_speedup}

As noted by \citep{blum2010approximate,biau2015new}, ABC is equivalent
to a kernel-CDE. More specifically, it can be
seen as a ``nearest-neighbors'' kernel-CDE (NN-KCDE) defined by
\begin{equation}
  \label{eq:kernel-nn-cde}
  \widehat{f}_{\text{nn}}(\theta \mid \x) = \frac{1}{k} \sum_{i = 1}^{k} K_{h}(\rho(\theta, \theta_{s_{i}(\x)})),
\end{equation}
where $s_{i}(\x)$ represents the index of the $i$th nearest neighbor to the target point $\x$ in covariate space, and we compute the conditional density of $\theta$ at $\x$ by applying a kernel smoother $K_{h}(\cdot)$ with bandwidth $h$ to the $k$ points closest to $x$.

For a given set of generated data, the above is equivalent to selecting the
ABC threshold $\epsilon$ as the $k/n$-th quantile of the observed
distances. This is commonly used in practice as it is more convenient
than determining $\epsilon$ a priori. However, as pointed out by \citet{biau2015new} (Section 4; remark 1), there is currently no good methodology to select both $k$ and $h$ in an ABC k-nearest neighbor estimate.

Given the connection between ABC and NN-KCDE, we propose to use our 
 surrogate loss
to tune the estimator; selecting $k$ and $h$ such that they minimize the
estimated surrogate loss in Equation~\ref{eq::lossEstimate}. In this sense, we are selecting the ``optimal''
ABC parameters after generating some of the data: having generated
10,000 points, it may turn out that we would have preferred a smaller
tolerance level $\epsilon$
 and that only using the closest 1,000 points would
better approximate the posterior.

{\bf Example with Normal Posterior.} 
To demonstrate the effectiveness of our surrogate loss in reducing the number of simulations, we draw data $X_1,\ldots,X_{5}|\mu \overset{iid}{\sim}  \operatorname{Normal}(\mu,0.2^2)$ where $\mu \sim  \operatorname{Normal}(1, 0.5^2)$. 
We examine the role of ABC thresholds by fitting the model for several values of the
threshold with observed data $\x_{0} = \left\{-0.5, -0.25, 0.0, 0.25, 0.5\right\}$. (A similar example with a two-dimensional normal distribution
can be found in the Appendix.) For each threshold, we perform rejection sampling until we
retain $B=1000$ ABC points. We select ABC thresholds to fix the acceptance rate
of the rejection sampling. Those acceptance rates are then used in
place of the actual tolerance level $\epsilon$ for easier comparison.

\begin{figure}[H]
  \centering
  \includegraphics[scale=0.39]{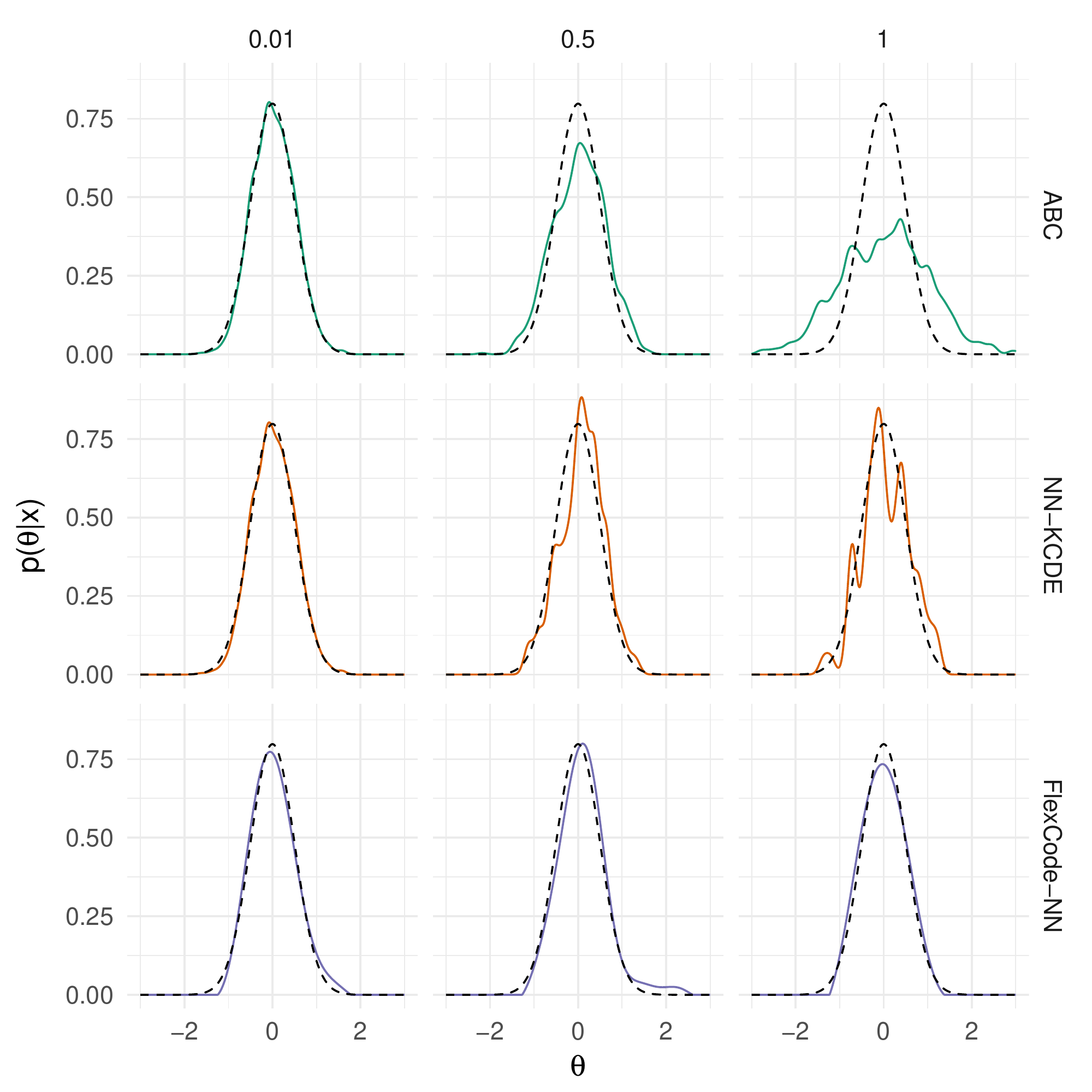}
   \includegraphics[scale=0.39]{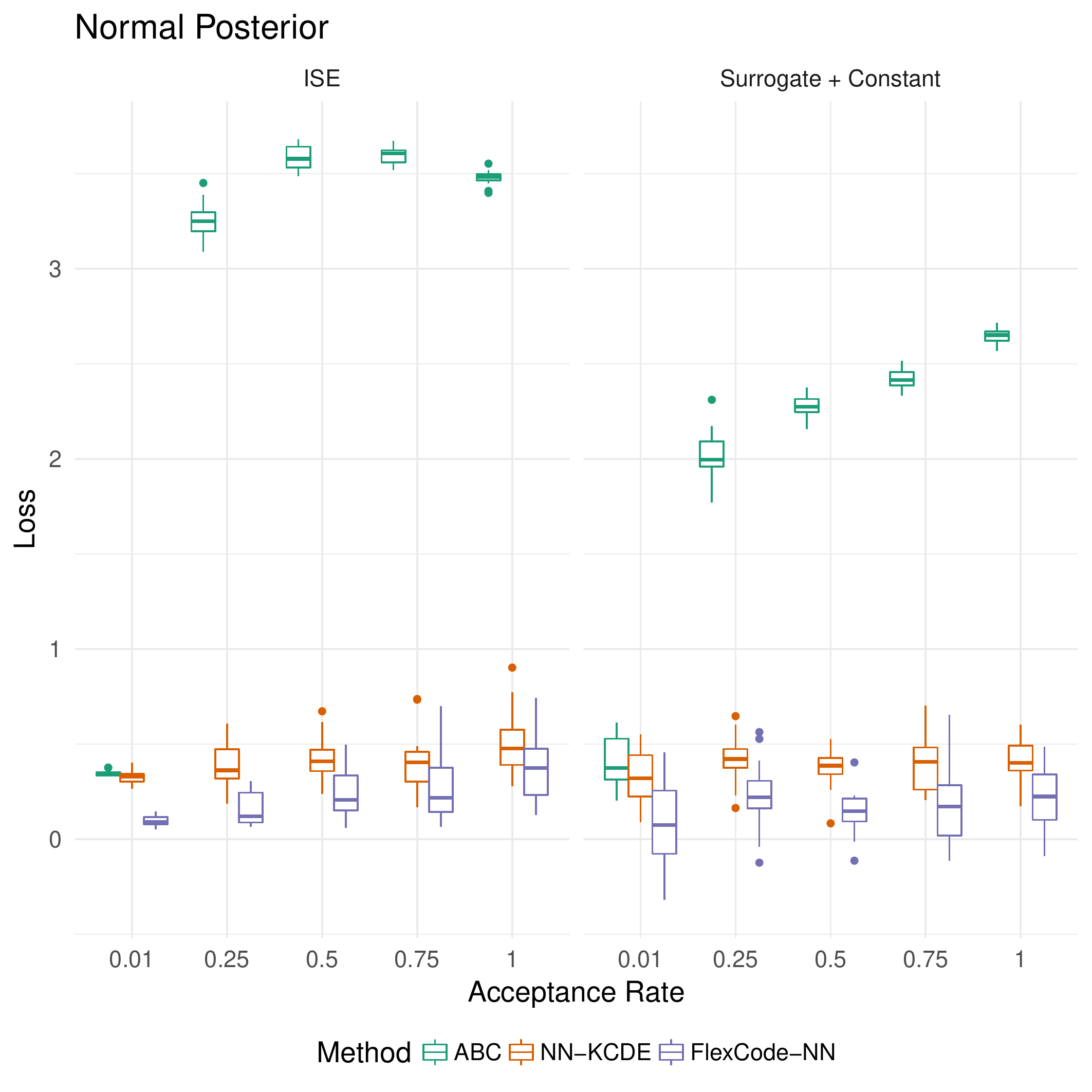}
    \vspace{-6mm}
  \caption{\footnotesize {\em Left:} Density estimates for normal posterior using ABC sample of
    varying acceptance rates (0.01, 0.5, and 1). As the acceptance rate (or, equivalently, the ABC tolerance level) decreases, 
    the ABC posterior approaches the true posterior. Both 
      NN-KCDE and FlexCode-NN approximate the posterior well for all acceptance
    rates, even for an acceptance rate of 1 which corresponds to no ABC threshold. {\em Right:} True integrated squared error (ISE) loss and estimated surrogate loss for normal posterior using ABC
        sample of varying acceptance rates. We need to decrease the acceptance rate considerably to attain a small loss for ABC. On the other hand,  the losses for
        NN-KCDE  and FlexCode-NN are 
        small for all thresholds.}
  \label{fig::normal-densities}
\end{figure}

\comment{consider the example of estimating the
posterior of a normal distribution. We draw the data
$X_{i} \sim \operatorname{Normal}(0.5, 1.0)$ and use the conjugate
prior $\mu \sim \operatorname{Normal}(0, 1)$. This results in the true
posterior
\begin{equation*}
  \mu \mid X_1,\ldots,X_n \sim \operatorname{Normal}\left(\frac{n}{n + 1}\bar{x}, \frac{1}{1 + n}\right), 
\end{equation*}
where $n$, the number of observations, is set at 10 in our simulations.

We use the sufficient statistic of the sample mean as our statistic
and the Euclidean norm as the distance function.}

The left panel of Figure \ref{fig::normal-densities} shows examples of posterior densities for varying
acceptance rates. For the highest acceptance rate of 1 (corresponding to the ABC tolerance level $\epsilon \rightarrow \infty$), the ABC
posterior (top left) is the prior distribution and thus a poor estimate. In contrast, the two ABC-CDE methods (FlexCode-NN and NN-KCDE) have a decent performance even at an acceptance rate of 1; more generally, they perform well at a higher acceptance rate than standard ABC. 

To corroborate this qualitative look, we examine the loss for each
method. The right panel of Figure \ref{fig::normal-densities} plots the true and surrogate
losses against the acceptance rate for the given methods. As seen in
Section \ref{sec::estPostModelSelec}, the surrogate loss provides the same
conclusion as the (unavailable in practice) true loss. As the
acceptance rate decreases, the ABC realizations 
more closely approximate the true posterior and the ABC estimate of the posterior improves. The main result is 
that NN-KCDE and FlexCode-NN have roughly constant performance over
all values of the acceptance rate. As such, we could generate {\em only}
1,000 realizations of the ABC sample at an acceptance rate of 1 and achieve {\em similar} result
as standard ABC generating 100,000 values at an acceptance rate of 0.01.

There are two different sources of improvement: the first exhibited by
NN-KCDE amounts to selecting the ``optimal'' ABC parameters
$k$ and $h$ using surrogate loss. However, as FlexCode-NN performs slightly better than NN-KCDE for the same sample, there is an additional
improvement in using CDE methods other than kernel smoothers; this difference becomes more pronounced for high-dimensional and complex data  (see \citealt{IzbickiLeeFlexCode} for examples of when traditional kernel smoothers fail).

\comment{If simulations are sufficiently expensive, ABC rejection-sampling and
even iterative sampling, such as SMC, will be computationally
impractical. In this case, our results at high acceptance rates are
promising: at the cost of far fewer simulated values, we can achieve
similar performance as an ABC sample with smaller acceptance rate and
consequently a larger number of simulations. We can also augment methods
designed to accelerate convergence, taking the output of a SMC
algorithm as the training sample for our CDE methods.

A similar example with a two-dimensional normal distribution
can be found in the Appendix. The conclusions are similar.}

\section{Conclusions}
\label{sec::final}

In this work, we have demonstrated three ways in which our conditional estimation framework 
can improve upon approximate Bayesian computational methods for next-generation complex data and simulations. 

First, realistic simulation models are often such that the computational cost of generating a single sample 
is large, making lower acceptance ratios unrealistic. 

Secondly, our ABC-CDE framework allows one to compare ABC and related methods in a principled way, making it possible 
to pick the best method for a given data set without knowing the true posterior.
Our approach is based on a surrogate loss function and data splitting. 
 We note that a related cross-validation procedure to choose the tolerance level $\epsilon$ in ABC has been proposed
by \citet{abcPackage}, albeit using a loss function that is appropriate for 
point estimation only. 

Finally, when dealing with complex models, it is often difficult to know exactly what summary statistics would be appropriate for ABC. Nevertheless, the practitioner can usually make up a list of a large but redundant number of candidate statistics, including statistics
		generated with automatic methods. As our results show, FlexCode-RF (unlike ABC) is robust to irrelevant statistics. 
		Moreover,
	FlexCode, in combination with RF for regression, offers a way of evaluating the importance of each summary statistic
	in estimating the full posterior distribution; hence, these importance scores could be used to choose relevant summary statistics for ABC and any other method used to estimate posteriors.
	
In brief, there are really two estimation problems in ABC-CDE: The first is that of estimating $f(\theta | \x_o)$. ABC-CDE starts with a rough approximation from an ABC sampler and then directly
 estimates the conditional density exactly at the point $\x = \x_o$ using a nonparametric conditional density estimator. The second is that of estimating the integrated squared error loss
 (Eq. \ref{eq::trueLoss}). Here we propose a surrogate loss that weights all points in the ABC posterior sample equally, but a weighted surrogate loss could potentially return more accurate estimates of the ISE. For example, 
 Figures 9 (left)  and
 10 in the Appendix
 show that NN-KCDE perform better than 
 ABC post-processing techniques. The current estimated loss, however,  cannot identify a difference in ISE loss between NN-KCDE and ``Blum'' because of the rapidly shifting posterior in the vicinity of $\x=\x_0$.


\vspace{3mm}

{\small \noindent \textbf{Acknowledgments.}
We are grateful to  Rafael Stern
 for his insightful comments on the manuscript. We would also like to thank Terrence Liu, Brendan McVeigh,  and Matt Walker for the NFW simulation code and to Michael Vespe for help with the weak lensing simulations in the Appendix. This work was partially supported by  NSF DMS-1520786,
Funda\c{c}\~ao de Amparo \`a Pesquisa do Estado de S\~ao Paulo} (2017/03363-8)
and
CNPq (306943/2017-4). 

\vspace{3mm}

{\small \noindent {\textbf{Links to Nonparametric Software Optimized for CDE:}
\begin{itemize}
\item  FlexCode: \url{https://github.com/rizbicki/FlexCoDE};  \url{ https://github.com/tpospisi/flexcode}
\item NN-KCDE: \url{https://github.com/tpospisi/NNKCDE} (see Appendix D)
\item RF-CDE: \url{https://github.com/tpospisi/rfcde} \citep{RFCDE_arxiv}
\end{itemize}

\bibliographystyle{Chicago}

{ \footnotesize 
\bibliography{paper}
}

\appendix

\section{Proofs} 

\subsection{Results on the surrogate loss}

\begin{thm}
\label{Eq::LossBias}
Assume that, for every $\theta \in \Theta$, $g_\theta(\x):= (\widehat{f}(\theta|\x)-f(\theta|\x))^2$
satisfies the  H\"{o}lder condition of order $\beta$ with a constant $K_\theta$\footnote{That is, there exists a constant $K_\theta$ such that for every $\x,\vec{y} \in \Re^d$  $|g_\theta(\x)-g_\theta(\vec{y})|\leq K_\theta (d(\x,\vec{y}))^\beta$. } such that $K_H:=\int K_\theta d\theta < \infty$. Then
$$|L_{\x_o}^\epsilon(\widehat{f},f)-L_{\x_o}(\widehat{f},f)| = K_H \epsilon^\beta = O(\epsilon^\beta)$$
\end{thm}

\begin{proof}
First, notice that
$$L_{\x_o}(\widehat{f},f)= \int g_\theta(\x_o) d\theta \int\frac{f(\x) \I(d(\x,\x_o)< \epsilon)}{\P(d(\X,\x_o)<\epsilon)} d\x=\int \int g_\theta(\x_o) \frac{f(\x) \I(d(\x,\x_o)< \epsilon)}{\P(d(\X,\x_o)<\epsilon)} d\x d\theta.$$
It follows that  
\begin{align*}
|L_{\x_o}^\epsilon(\widehat{f},f)-L_{\x_o}(\widehat{f},f)| &=  \left| \int \left( \int g_\theta(\x) \frac{f(\x) \I(d(\x,\x_o)< \epsilon)}{\P(d(\X,\x_o)<\epsilon)} d\x  - \int g_\theta(\x_o) \frac{f(\x) \I(d(\x,\x_o)< \epsilon)}{\P(d(\X,\x_o)<\epsilon)} d\x   \right) d\theta\right| \\
& = \left| \int \left( \int (g_\theta(\x) -g_\theta(\x_o))\frac{f(\x) \I(d(\x,\x_o)< \epsilon)}{\P(d(\X,\x_o)<\epsilon)} d\x   \right) d\theta\right| \\
&\leq \int \left( \int \left| g_\theta(\x) -g_\theta(\x_o)\right|\frac{f(\x) \I(d(\x,\x_o)< \epsilon)}{\P(d(\X,\x_o)<\epsilon)} d\x   \right) d\theta \\
&\leq \int \left( \int K_\theta d(\x,\x_o)^\beta \frac{f(\x) \I(d(\x,\x_o)< \epsilon)}{\P(d(\X,\x_o)<\epsilon)} d\x   \right) d\theta \\
&\leq \int K_\theta \epsilon^\beta \left( \int  \frac{f(\x) \I(d(\x,\x_o)< \epsilon)}{\P(d(\X,\x_o)<\epsilon)} d\x   \right) d\theta \\
& = \epsilon^\beta   \int K_\theta  1 d\theta = K_H \epsilon^\beta
\end{align*}

\end{proof}

\begin{align}
\label{eq::LHat}
L_{\x_o}^\epsilon(\widehat{f},f) &= \notag \\
&\int \! \int \widehat{f}^2(\theta|\x) \frac{f(\x) \I(d(\x,\x_o)< \epsilon)}{\P(d(\X,\x_o)<\epsilon)}  d\theta d\x -2 
\int \! \int  \widehat{f}(\theta|\x)f(\theta|\x) \frac{f(\x) \I(d(\x,\x_o)< \epsilon)}{\P(d(\X,\x_o)<\epsilon)} d\theta d\x
+ K_f \notag \\
&  =\E_{\X'}\left[\int \widehat{f}^2(\theta|\X) d\theta \right]-2 \E_{(\theta',\X')}\left[\widehat{f}(\theta|\X)\right]+K_f,
\end{align}

\begin{thm}
Let $K_f$ be as in Equation \ref{eq::LHat}. Under the assumptions of Theorem \ref{Eq::LossBias},
$$|\widehat{L}_{\x_o}^\epsilon(\widehat{f},f)+K_f-L_{\x_o}(\widehat{f},f)| = O(\epsilon^\beta)+O_P(1/\sqrt{B'})$$
\end{thm}

\begin{proof}
Using the triangle inequality,
\begin{align*}
|\widehat{L}_{\x_o}^\epsilon(\widehat{f},f)+K_f-L_{\x_o}(\widehat{f},f)| &\leq |\widehat{L}_{\x_o}^\epsilon(\widehat{f},f)+K_f-L_{\x_o}^\epsilon(\widehat{f},f)|+|L_{\x_o}^\epsilon(\widehat{f},f)-L_{\x_o}(\widehat{f},f)| \\
&= O(\epsilon^\beta)+O_P(1/\sqrt{B'}),
\end{align*}
where the last inequality follows from Theorem \ref{Eq::LossBias} and the fact that
$\widehat{L}_{\x_o}^\epsilon(\widehat{f},f)+K_f$ is an average of $B'$ iid random variables.
\end{proof}

\begin{Lemma} Assume there exists $M$ such that $|\widehat{f}(\theta|\x)|\leq M$ for every $\x$ and $\theta$.
Then
\label{lemma::estimateToEps}
 $$\P\left(|\widehat{L}_{\x_o}^\epsilon(\widehat{f},f)+K_f-L_{\x_o}^\epsilon(\widehat{f},f)| \geq 	\nu\right) \leq 2 e^{-\frac{B'\nu^2}{2(M^2+2M)^2}}$$
\end{Lemma}

\begin{proof}
Notice that
$$\widehat{L}_{\x_o}^\epsilon(\widehat{f},f)+K_f-L_{\x_o}^\epsilon(\widehat{f},f) = \frac{1}{B'}\sum_{k=1}^{B'} W_k-\E[W_1],$$
where
$W_k= \int \widehat{f}^2(\theta|\X'_k) d\theta  -2 \widehat{f}(\Theta'_k|\X'_k)$, with $W_1,\ldots,W_{B'}$ iid.
The conclusion follows from Hoeffding's inequality and the fact that
$|W_k| \leq |\int \widehat{f}^2(\theta|\X'_k) d\theta  -2 \widehat{f}(\Theta'_k|\X'_k)|
\leq M^2+2M.$
\end{proof}

\begin{Lemma} 
\label{lemma::estimateToTrue}
Under the assumptions of Lemma \ref{lemma::estimateToEps} and if
$g_{\theta}(\x):= (\widehat{f}(\theta|\x)-f(\theta|\x))^2$
satisfies the  H\"{o}lder condition of order $\beta$ with constants $K_\theta$
such that $ K_H:=\int K_\theta d\theta < \infty$, 
 $$\P\left(|\widehat{L}_{\x_o}^\epsilon(\widehat{f},f)+K_f-L_{\x_o}(\widehat{f},f)| \geq K_H \epsilon^\beta + \nu\right) \leq 2 e^{-\frac{B'\nu^2}{2(M^2+2M)^2}},$$
\end{Lemma}

\begin{proof}
Notice that
\begin{align*}
|&\widehat{L}_{\x_o}^\epsilon(\widehat{f},f)+K_f-L_{\x_o}(\widehat{f},f)| - K_H \epsilon^\beta \\
&= 
|\widehat{L}_{\x_o}^\epsilon(\widehat{f},f)+K_f-L_{\x_o}^\epsilon(\widehat{f},f)+L_{\x_o}^\epsilon(\widehat{f},f)-L_{\x_o}(\widehat{f},f)| - K_H \epsilon^\beta  \\
&\leq 
|\widehat{L}_{\x_o}^\epsilon(\widehat{f},f)+K_f-L_{\x_o}^\epsilon(\widehat{f},f)|+|L_{\x_o}^\epsilon(\widehat{f},f)-L_{\x_o}(\widehat{f},f)| - K_H \epsilon^\beta  \\
&\leq 
|\widehat{L}_{\x_o}^\epsilon(\widehat{f},f)+K_f-L_{\x_o}^\epsilon(\widehat{f},f)|,
\end{align*}
where the last line follows from Theorem \ref{Eq::LossBias}.
It follows that 
$$ |\widehat{L}_{\x_o}^\epsilon(\widehat{f},f)+K_f-L_{\x_o}(\widehat{f},f)| \geq K_H \epsilon^\beta + \nu \Rightarrow |\widehat{L}_{\x_o}^\epsilon(\widehat{f},f)+K_f-L_{\x_o}^\epsilon(\widehat{f},f)| \geq 	\nu.$$ 
The conclusion follows from Lemma \ref{lemma::estimateToEps}.
\end{proof}

\begin{thm}
\label{thm::estimateToTrueUnion}
Let $\mathcal{F}=\{\widehat{f}_1,\ldots,\widehat{f}_m\}$ be a   set of estimators of $f(\theta|\x_o)$.
Assume there exists $M$ such that $|\widehat{f}_i(\theta|\x)|\leq M$ for every $\x$, $\theta$, and $i=1,\ldots,m$. \footnote{Such
assumptions hold if the $\widehat{f}_i$'s are obtained
via FlexCode with bounded basis functions (e.g., Fourier basis)
or a kernel density estimator on the ABC samples.}
Moroever, assume that
for every $\theta \in \Theta$, $g_{i,\theta}(\x):= (\widehat{f}_i(\theta|\x)-f(\theta|\x))^2$
satisfies the  H\"{o}lder condition of order $\beta$ with constants $K_\theta$
such that $ K_H:=\int K_\theta d\theta < \infty$. Then,
 $$\P\left(\max_{\widehat{f} \in \mathcal{F}} |\widehat{L}_{\x_o}^\epsilon(\widehat{f},f)+K_f-L_{\x_o}(\widehat{f},f)| \geq K_\epsilon \epsilon^\beta+	\nu\right) \leq  2m e^{-\frac{B'\nu^2}{2(M^2+2M)^2}}.$$
\end{thm}

\begin{proof}
The theorem follows from Lemma \ref{lemma::estimateToTrue}  and the union bound.
\end{proof}

\begin{Cor}
Let
$\widehat{f}^* := \arg \min_{\widehat{f} \in \mathcal{F}} \widehat{L}_{\x_o}^\epsilon(\widehat{f},f)$
be the best estimator in $\mathcal{F}$ according to the estimated surrogate loss, and
 let $f^*=\arg \min_{\widehat{f} \in \mathcal{F}} L_{\x_o}(\widehat{f},f)$
be the best estimator in $\mathcal{F}$ according to the true loss.
Then, under the assumptions from Theorem \ref{thm::estimateToTrueUnion}, with probability at least $1-2m e^{-\frac{B'\nu^2}{2(M^2+2M)^2}}$,
$$ L_{\x_o}(\widehat{f}^*,f)\leq  L_{\x_o}(f^*,f) +2 (K_H\epsilon^\beta+	\nu).$$
\end{Cor}

\begin{proof}
From Theorem \ref{thm::estimateToTrueUnion},
with probability at least $1-2m e^{-\frac{B'\nu^2}{2(M^2+2M)^2}}$
\begin{align*}
L_{\x_o}(\widehat{f}^*,f)-&L_{\x_o}(f^*,f)= \\
&L_{\x_o}(\widehat{f}^*,f)-(\widehat{L}_{\x_o}^\epsilon(\widehat{f}^*,f)+K_f) \\
&+ (\widehat{L}_{\x_o}^\epsilon(\widehat{f}^*,f)+K_f)-(\widehat{L}_{\x_o}^\epsilon(f^*,f)+K_f)  \\
&+(\widehat{L}_{\x_o}^\epsilon(f^*,f)+K_f) - L_{\x_o}(f^*,f) \\
&\leq 2 (K_H\epsilon^\beta+	\nu),
\end{align*}
where the inequality follows from the fact that, by definition,
$(\widehat{L}_{\x_o}^\epsilon(\widehat{f}^*,f)+K_f)-(\widehat{L}_{\x_o}^\epsilon(f^*,f)+K_f) <0$
and $L_{\x_o}(\widehat{f},f)-(\widehat{L}_{\x_o}^\epsilon(\widehat{f},f)+K_f) \leq K_H\epsilon^\beta+	\nu$
for every $\widehat{f} \in \mathcal{F}$.
\end{proof}

\subsection{Results on summary statistics selection}

\begin{Assumption}[Smoothness in $\theta$ direction]
\label{label:assumpSmooth}
 \label{assump-sobolevZ} $\forall \vec{x} \! \in \! \mathcal{X}$, 
$f(\theta|\x) \! \in \! W_{\phi}(s_\vec{x},c_\vec{x}),$ 
 where $f(\theta|\x)$ is viewed as a function of $\theta$, and $s_\vec{x}$ and $c_\vec{x}$ are such that 
$\inf_\vec{x} s_\vec{x}\overset{\mbox{\tiny{def}}}{=}\beta>\frac{1}{2}$ and 
$\int c^2_\vec{x} d\x <\infty$.
\end{Assumption}

\begin{Lemma}
\label{lemma::sobolev}
Let $\x=(x_1,\ldots,x_d)$ and $\x'=(x_1,\ldots,x_j', \ldots, x_d)$.
Then, for every $\x$ and $x_j$, $g_{\x,x_j}(\theta):=f(\theta|\x)-f(\theta|\x_j') \in W_{\phi}(\beta,c_\x^2+c_{\x_j'}^2+2\sqrt{c_\x^2c_{\x_j'}^2})$.
\end{Lemma}

\begin{proof}
First we expand $f(\theta|\x)$ 
and $f(\theta|\x_j')$  in the basis $(\phi_i)_i$.
We have that
\begin{align*}
g_{\x,x_j}(\theta)=f(\theta|\x)-f(\theta|\x'_j) =\sum_{i \geq 0} (\beta_i(\x)-\beta_i(\x'_j))\phi_i(\theta).
\end{align*}
Now, using  Cauchy-Schwarz inequality, the expansion coefficients satisfy
\begin{align*}
\sum_{i\geq 1} i^{2\beta} &(\beta_i(\x)-\beta_i(\x'_j))^2   \\
&=\sum_{i\geq 1} i^{2\beta} (\beta_i(\x))^2
+\sum_{i\geq 1} i^{2\beta} (\beta_i(\x'_j))^2+2\sum_{i\geq 1} i^{2\beta} \beta_i(\x)\beta_i(\x'_j) \\
&\leq c_\x^2+c_{\x'_j}^2+2\sqrt{\left(\sum_{i\geq 1} i^{2\beta} (\beta_i(\x))^2\right) \left(\sum_{i\geq 1} i^{2\beta} (\beta_i(\x_j'))^2\right)} \\
&\leq c_\x^2+c_{\x'_j}^2+2\sqrt{c_\x^2c_{\x'_j}^2},
\end{align*}
where the last inequality follows from Assumption \ref{assump-sobolevZ}.
\end{proof}

\begin{prop} Under Assumption \ref{label:assumpSmooth},
$$ r_j =  \sum_{i=1}^I r_{i,j} + O\left(I^{-2\beta}\right)$$ \label{prop::stats_relevance}
\end{prop}
\begin{proof} 
Because $\beta_i(\x)-\beta_i(\x'_j)$ are the expansion coefficients of 
$f(\theta|\x)-f(\theta|\x'_j)$ on the basis $(\phi_i)_i$, 
it follows from Lemma \ref{lemma::sobolev} (see appendix) that
\begin{align*}
\sum_{i \geq I} I^{2\beta}\left(\beta_i(\x)-\beta_i(\x_j')\right)^2 \leq \sum_{i \geq I} i^{2\beta}\left(\beta_i(\x)-\beta_i(\x_j')\right)^2 \leq  c_\x^2+c_{\x'_j}^2+2\sqrt{c_\x^2c_{\x'_j}^2}.
\end{align*}
Hence,
\begin{align}
\label{eq::bias}
\sum_{i \geq I} r_{i,j}=\sum_{i \geq I} \int \! \int \left(\beta_i(\x)-\beta_i(\x_j')\right)^2 d\x dx_j' \leq \frac{ K}{I^{2\beta}}
= O(I^{-2 \beta}).
\end{align}
Because $f(\theta|\x)-f(\theta|\x'_j) =\sum_{i \geq 0} (\beta_i(\x)-\beta_i(\x'_j))\phi_i(\theta)$
and the basis $(\phi_i)_i$ is orthonormal, we have that
\begin{align} 
\label{eq::expansion}
r_j= \int \! \int \sum_{i \geq 0} (\beta_i(\x)-\beta_i(\x'_j))^2  d\x dx_j' 
= \sum_{i \geq 0} r_{i,j}.
\end{align}
The final result follows from putting Equations \ref{eq::bias}
 and \ref{eq::expansion} together.
 \end{proof}

\section{Mean of a Gaussian with unknown precision}

In this section,
we repeat the experiments of Section 3.1 of the paper, but in the case
		$X_1,\ldots,X_{20}|(\mu,\tau) \overset{iid}{\sim} \mbox{N}(\mu,1/\tau)$, $(\mu,\tau) \sim \mbox{Normal-Gamma}(\mu_0,\nu_0,\alpha_0,\beta_0)$. We set $\mu_0=0$, $\alpha_0=2$, $\beta_0=50$ and repeat the experiments for $\nu_0$ in an equally spaced grid with ten  values between 0.001 and 1.

		    \begin{figure}[H]
		    	\centering
		    	\subfloat[]{  \includegraphics[page=1,scale=0.29]{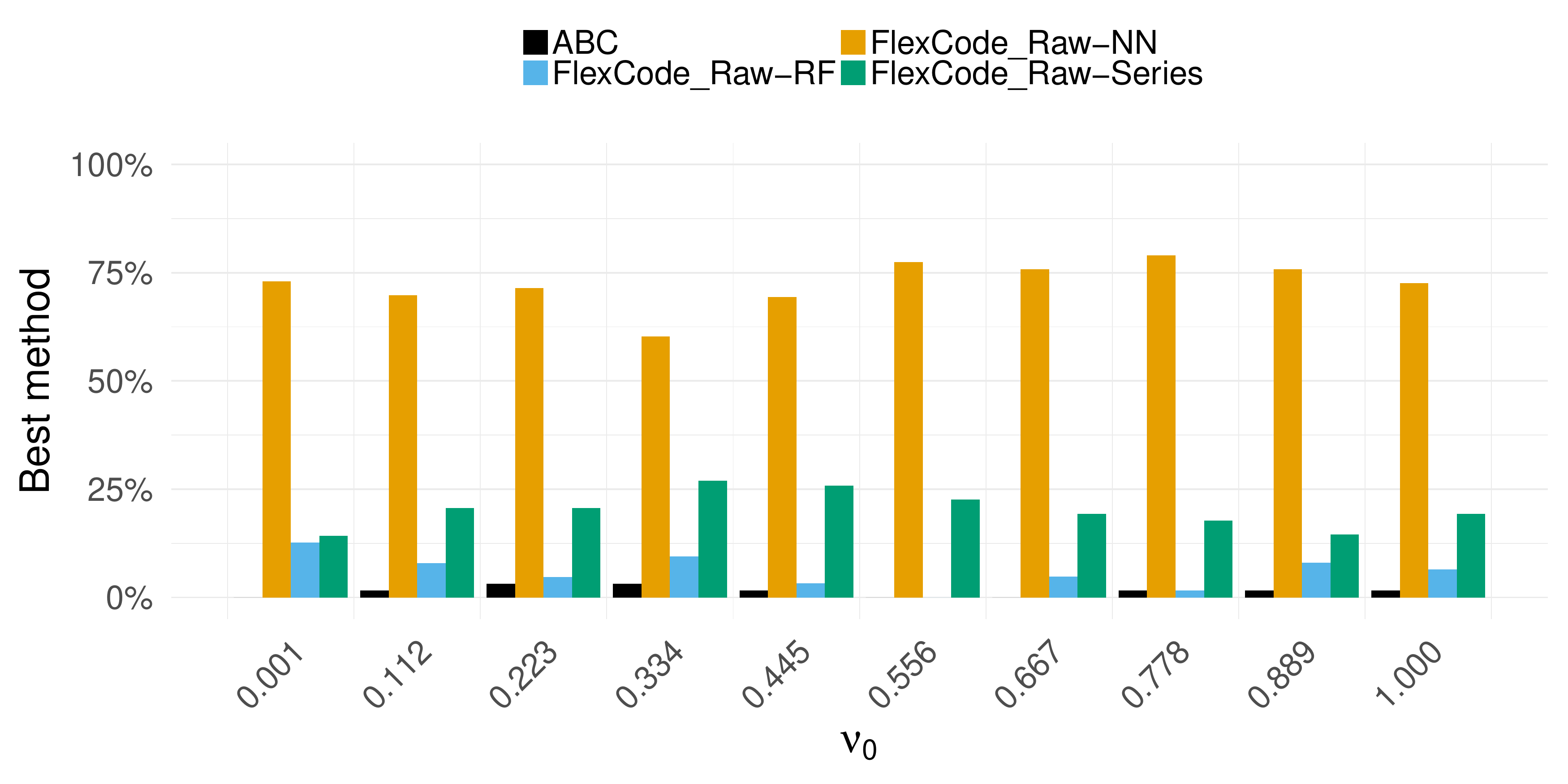}} 
		    	\subfloat[]{  \includegraphics[page=1,scale=0.29]{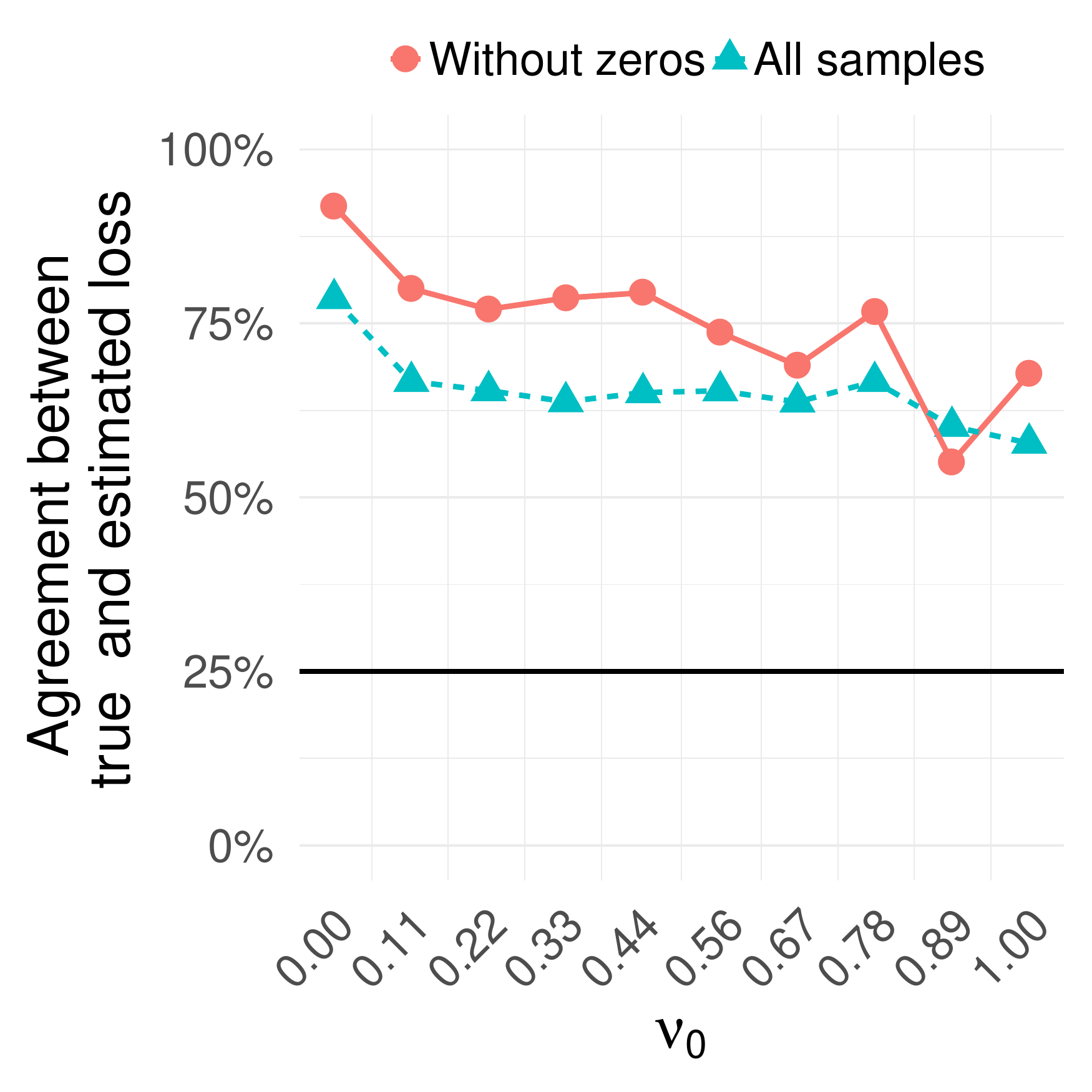}}  \\
		    	\subfloat[]{  \includegraphics[page=2,scale=0.29]{figures/gauss_unknownVar_reduced.pdf}} 
		    	\subfloat[]{  \includegraphics[page=3,scale=0.29]{figures/gauss_unknownVar_scatter_reduced.pdf}}  
		    	\vspace{-3mm}
		    	\caption{\footnotesize  CDE and method selection results for scenario 3 (mean of a Gaussian with unknown precision).  {\em Left:}	 Panels (a) and (c) show that the NN version of FlexCode yield better estimates of the posterior density $f(\theta|\x_o)$ than the competing methods. {\em Right:} Panels (b) and (d) indicate that one by estimating the surrogate loss function can tell from the data which method is better for the problem at hand. The horizontal line in panel (b) represents the behavior of a random selection.}
		    \end{figure}

		      \begin{figure}[H]
		      	\centering
		      	\subfloat[]{  \includegraphics[page=2,scale=0.29]{figures/summary_setting2.pdf}} \\
		      	\subfloat[]{  \includegraphics[page=1,scale=0.23]{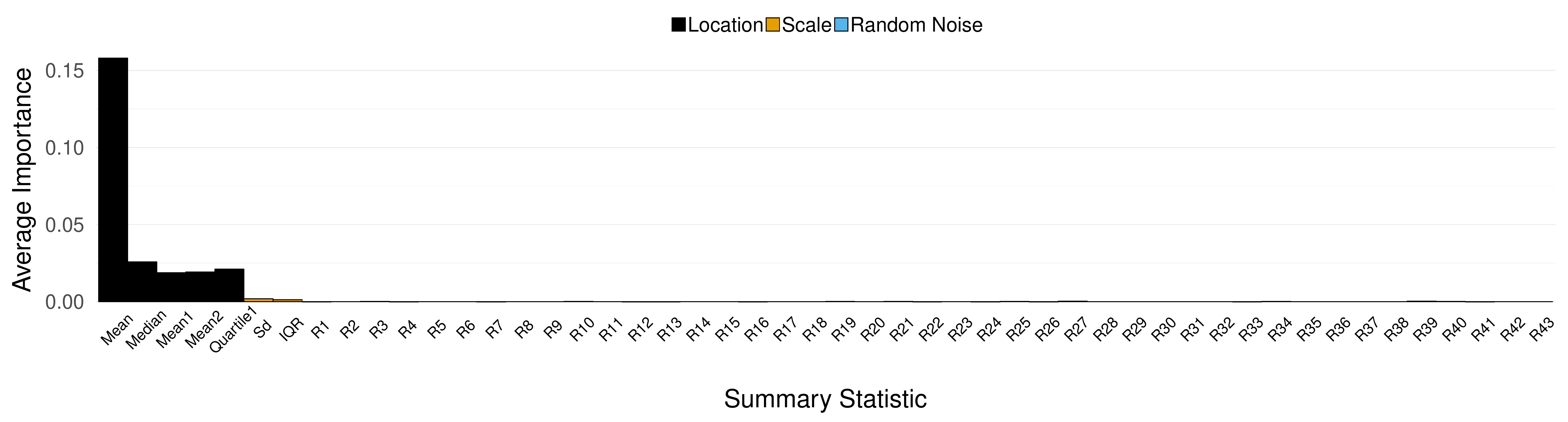}}  
		      	\vspace{-3mm}
		      	\caption{\footnotesize      Summary statistic selection for scenario 3 (mean of a Gaussian with unknown precision).  
		      	  	Panel (a) shows the performance of ABC and FlexCode-RF using different sets of summary statistics. The ABC estimates of the posteriors rapidly deteriorate when adding other statistics than the location statistics 1-5 ({\em top left}), whereas nuisance statistics do not decrease the performance of FlexCode-RF significantly ({\em top right}). Panel (b), furthermore, shows that FlexCode-RF identifies the location statistics (entries 1-5) as key variables and assigns them a high average importance score.} 
		      	\label{fig::toyExampleSummary2}
		      \end{figure}

\section{Application I: Estimating a Galaxy's Dark Matter Density Profile}
      \label{sec::sculptor-data}
      
      Next we consider more complex simulations.  
      The $\Lambda$CDM  (Lambda cold dark matter) model is frequently referred to as the standard model of Big Bang cosmology \citep{liddle2015introduction}; it is the simplest model that contains assumptions consistent with observational and theoretical knowledge of the Universe.
            
      The $\Lambda$CDM model predicts that the dark
      matter profile of a galaxy in the absence of baryonic effects can be parameterized by the
      Navarro-French-White (NFW) model \citep{navarro1996structure}. Given
      an observed galaxy, such as the Sculptor dwarf spheroidal galaxy, we
      wish to constrain the parameters of the NFW model. To begin we will only
      consider a single parameter, the critical energy $E_{c}$ (\citealt{strigari2017dynamical}, Equation 15), and set all
      other parameters at commonly accepted estimates; see Section      \ref{sec::detailsNFW} for details.
      
The observed data $\x_0$ are velocities and coordinates of 200 stars in a galaxy, here simulated so as to follow the NFW model.\footnote{The simulations are written by Mao-Shen (Terrence) Liu and rely on an MCMC sampling scheme; the details are outlined in~\citealt{LiuWalker2018}.}
      To perform ABC we define
      the distance function as the $\ell^2$ norm between bivariate
      kernel density estimates of the joint distribution of the
      velocity and distance from the center. The same distance
      function will be used by FlexCode-NN and FlexCode-Series.
      Because the data are functional we also implement a third version of ABC-CDE based on
      FlexCode-Functional, where the coefficients in FlexCode are
      estimated via functional kernel regression
      \citep{ferraty2006nonparametric}.
      
      To assess the performance of the CDE methods we generate 1000
      test observations each with an ABC sample of 1000 accepted
      observations with an acceptance rate of 0.1. 
      We use the prior
      $E_{c} \sim U(0.01, 1.0)$.

      \begin{figure}[H]
        \centering
         \subfloat{\includegraphics[scale=0.5]{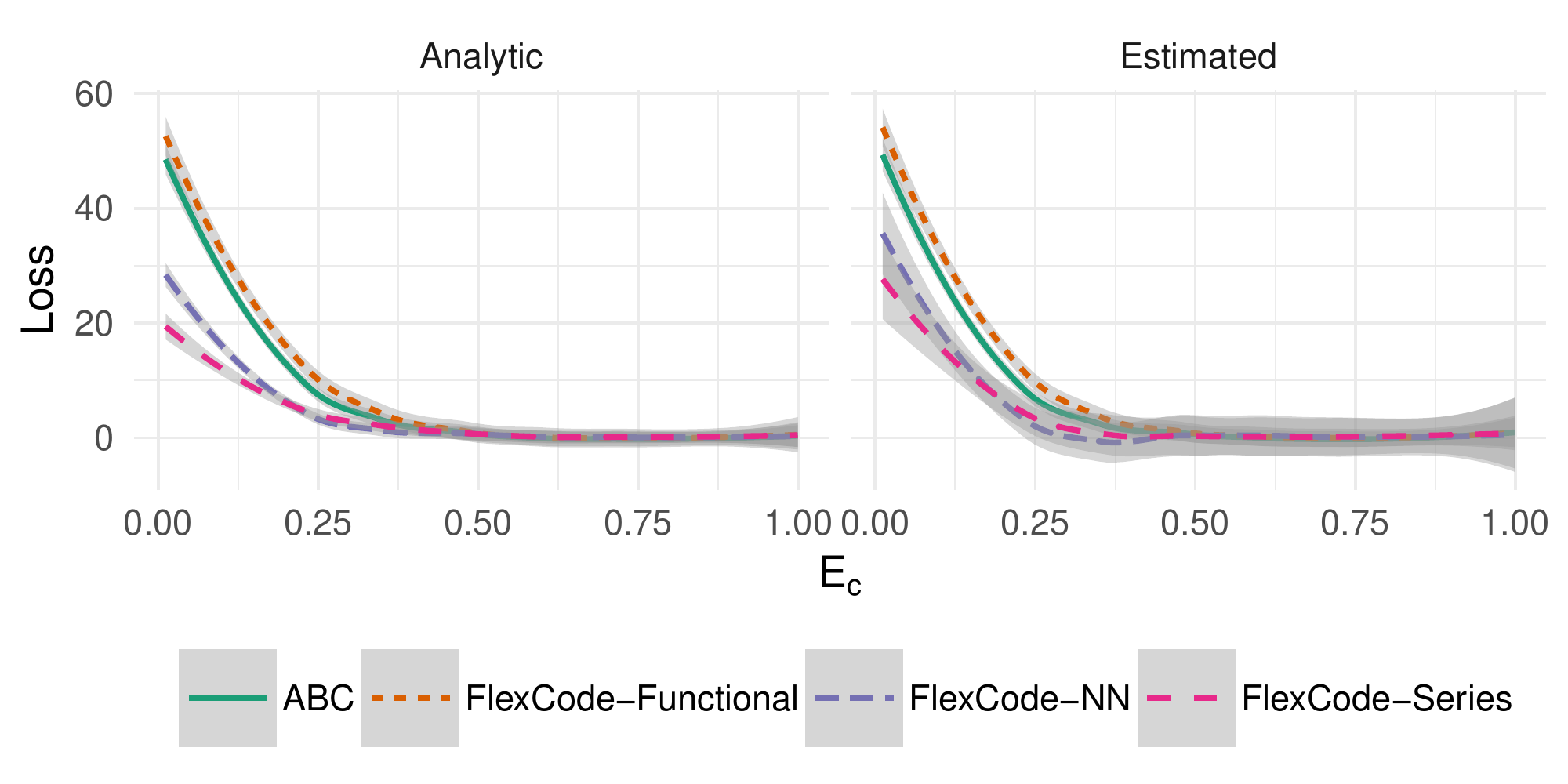}}
         \vspace{-2mm}
        \caption{  \footnotesize {\em Left:} True loss for each method in different parameter regions. {\em Right:} The estimated surrogate loss (here shifted with a constant 
        $\int \! f(\theta|\x_o)^2 d\theta$ for easier comparison) can be used to identify when our methods improve upon the ABC estimates.} 
        \label{fig::sculptor-losses}
      \end{figure}
\vspace{-6mm}
  
        Figure
        \ref{fig::sculptor-losses}, left, displays the true loss for
        each method. The plot indicates that, for at least some
        realizations with low true $E_{c}$, the estimates from the
        FlexCode estimators lead to better performance than ABC. Of most
        interest is that we reach similar conclusions with the estimated
        surrogate loss; see right plot. Thus the surrogate loss serves
        as a reasonable proxy for the true loss which in practical
        applications will be unavailable. 
        
\subsection{Additional details}        
\label{sec::detailsNFW}        
Under the NFW model the joint likelihood for the specific angular
momentum $J$ and specific energy $E$ factorizes independently

\begin{equation}
  \label{eq:NFW-likelihood}
  f(E, J) \propto g(J)h(E)
\end{equation}

with

\begin{equation}
  \label{eq:J-likelihood}
  g(J) = \begin{cases}
    [1 + (J / J_{\beta})^{-b}]^{-1} & b \le 0 \\
    1 + (J / J_{\beta})^{b} & b > 0
  \end{cases}
\end{equation}

and

\begin{equation}
  \label{eq:E-likelihood}
  h(E) = \begin{cases}
    E^{\alpha}(E^{q} + E_{c}^{q})^{d/q}(\Phi_{lim} - E)^{e} & E < \Phi_{lim} \\
    0 & E \ge \Phi_{lim}
  \end{cases}
\end{equation}

We can relate $E$ and $J$ to the observed values of position $r$ and
velocity $v$  as follows

\begin{eqnarray*}
  E &=& \frac{1}{2} v^{2} + \Phi_{s}(1 - \frac{\log(1 + \frac{r}{r_{s}})}{\frac{r}{r_{s}}}) \\
  J &=& vr\sin(\theta)
\end{eqnarray*}

We set the following constants at commonly accepted values
  in Table \ref{tab::values}
and focus only on estimating $E_{c}$.

\begin{table}
\centering
\caption{Parameter values used for the simulations
of the Galaxy's Dark Matter Density Profile Model as reported by \cite{strigari2017dynamical}.}
\begin{tabular}{|l|l|}
  \hline
  Parameter & Value \\
  \hline
  $\alpha$ & 2.0 \\
  d & -5.3 \\
  e & 2.5 \\
  $v_{max}$ & 21 \\
  $r_{max}$ & 1.5 \\
  $\phi_{s}$ & $(v_{max} / 0.465)^{2}$ \\
  $r_{s}$ & $r_{max} / 2.16$ \\
  $r_{lim}$ & 1.5 \\
  $\phi_{lim}$ & $\phi_{s}(1 - \frac{\log(1 + \frac{r_{lim}}{r_{s}})}{\frac{r_{lim}}{r_{s}}}$ \\
  b & -9.0 \\
  q & 6.9 \\
  $J_{\beta}$ & 0.086 \\
  \hline
\end{tabular}
\label{tab::values}
\end{table}

      Figure \ref{fig::sculptor-densities} displays examples of
      estimated posterior; each plot consists of an observed sample
      generated using a different true of $E_c$.
      
      \begin{figure}[H]
        \centering
         \includegraphics[scale=0.4]{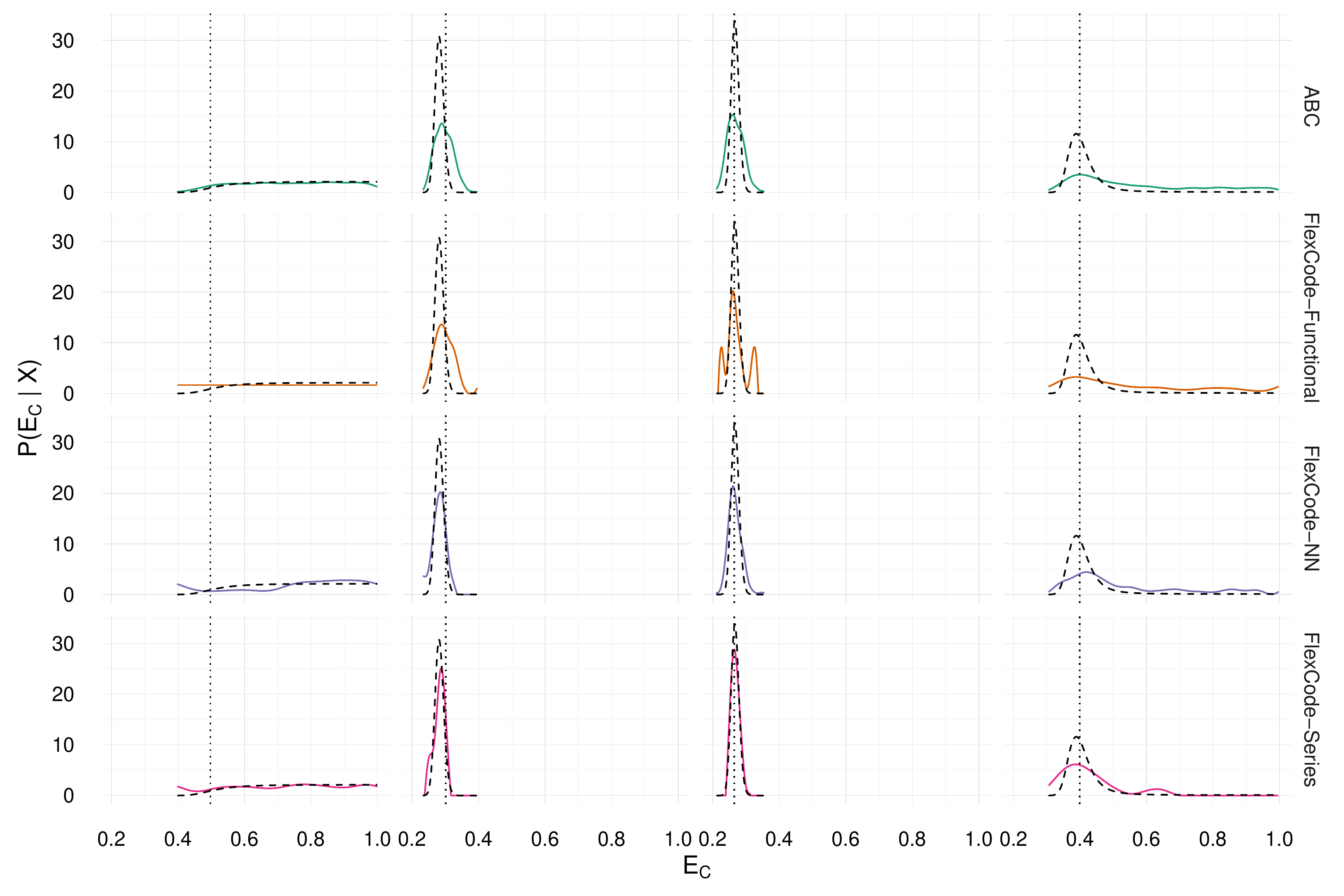}
        \caption{ \footnotesize   Sample posterior densities for simulated galaxy data; the dashed curve is the true posterior, and the vertical line indicates the true parameter value.}
        \label{fig::sculptor-densities}
      \end{figure}

\section{Fast Implementation of NN-KCDE}

NN-KCDE (or nearest-neighbors kernel CDE) is the usual kernel density estimate using only the points
closest in covariate space to the target point $\x$:
\begin{equation}
  \label{eq:kernel-nn-cde}
  \widehat{f}_{\text{nn}}(\theta \mid \x) = \frac{1}{k} \sum_{i = 1}^{k} K_{h}(\rho(\theta, \theta_{s_{i}(\x)})) ,
\end{equation}
where $s_{i}(\x)$ represents the index of the $i$th nearest neighbor to
$\x$. As mentioned in Section 4.1, NN-KCDE has a close connection with ABC in that,
for every choice of $k$ for a data set, there is a choice of $\epsilon$
for the accept-reject ABC algorithm that produces an equivalent
estimate of the posterior. With the CDE loss, we can choose $k$ to
minimize the loss and improve upon the naive pre-selected $\epsilon$
approach. For large $\epsilon$ we expect $k$ to be small to avoid bias
in the posterior estimate. As $\epsilon$ shrinks, larger values of $k$
will be selected to reduce the variance of the estimate.

To make this model selection procedure computationally feasible, we
need to be able to efficiently calculate the surrogate loss function. We examine
the two terms (Section 2.2, Equation 5) separately: The second term 
\begin{equation*}
  \sum_{i} \widehat{f}(\theta_{i} \mid \x_{i})
\end{equation*}
poses no difficulties as we simply plug in the kernel density
estimate. The first term
\begin{equation*}
  \sum_{i} \int \widehat{f}^2(\theta \mid \x_{i}) dz
\end{equation*}
is more difficult. Numerically integrating this integral is infeasible
especially as the number of validation samples increases.

Fortunately there is an analytic solution. We can express the integral
in terms of convolutions of the kernel function:
\begin{equation*}
  \int \widehat{f}^2(\theta \mid \x) dz = \frac{1}{k^{2}h} \sum_{i \in N_{k}(\x)}\sum_{j \in N_{k}(\x)} \int K(t) K\left(t - \frac{d_{i,j}}{h}\right) dt
\end{equation*}
with $d_{i,j}$ representing the pairwise distance between points
$\x_{i}$ and $\x_{j}$. In the Gaussian case we have the analytic solution
\begin{equation*}
  \int K(t) K(t - d) dt = \frac{1}{2\sqrt{\pi}} \exp\left(\frac{-d^{2}}{4}\right)
\end{equation*}

For other kernels we can work out the analytic solution as well, or,
if that proves intractable, we can approximate the function using
numerical integration.

For both terms we have nested calculations, in that we can reuse
computations for $k = k_{1} < k_{2}$ when calculating the estimated
loss for $k = k_{2}$. In this way, there is little additional
computational time in considering all settings for $k$ as opposed to
trying only a large value of $k$.

An implementation of this method is available at
\url{https://github.com/tpospisi/NNKCDE}.

\section{Two-Dimensional Normal.}
\label{sec:org0876fa9}

We can extend the normal example of Section 4 of the paper to
multiple dimensions with similar results. Given a two-dimensional multivariate normal
with fixed covariance \(\Sigma_{X}= I_{2}\), we put a normal conjugate prior on the mean
\(\mu \sim N(\mu_{0}= 0, \Sigma_{0} = I_{2})\). This results in the true posterior
\begin{equation*}
  \mu \mid X \sim N(\mu_{n}, \Sigma_{n}),
\end{equation*}
where
\begin{eqnarray*}
  \mu_{n} &=& \Sigma_{0}(\Sigma_{0} + \frac{1}{n}\Sigma_{X})^{-1}\bar{x} + \frac{1}{n}\Sigma_{X}(\Sigma_{0} + \frac{1}{n}\Sigma_{X})^{-1}\mu_{0} \\
  \Sigma_{n} &=& \frac{1}{n}\Sigma_{0}(\Sigma_{0} + \frac{1}{n}\Sigma_{X})^{-1} \Sigma_{X}
\end{eqnarray*}
As before we use the sufficient statistic of the sample mean as
our statistic and the Euclidean norm as the distance function.

Figure \ref{fig:org2bc0769} shows density estimates for ABC and NN-KCDE for different values of the acceptance rate. At
higher acceptance rates the ABC density estimate performs poorly,
reflecting the prior distribution rather than the posterior.
Eventually, with a suitably low acceptance rate the ABC density
approaches the posterior. Once again NN-KCDE achieves similar performance as standard ABC with 100000 simulations (at acceptance ratio 0.01) but using only 1000 ABC realizations (at acceptance ratio 1).
\begin{figure}[htbp]
\centering
\includegraphics[width=0.7\textwidth]{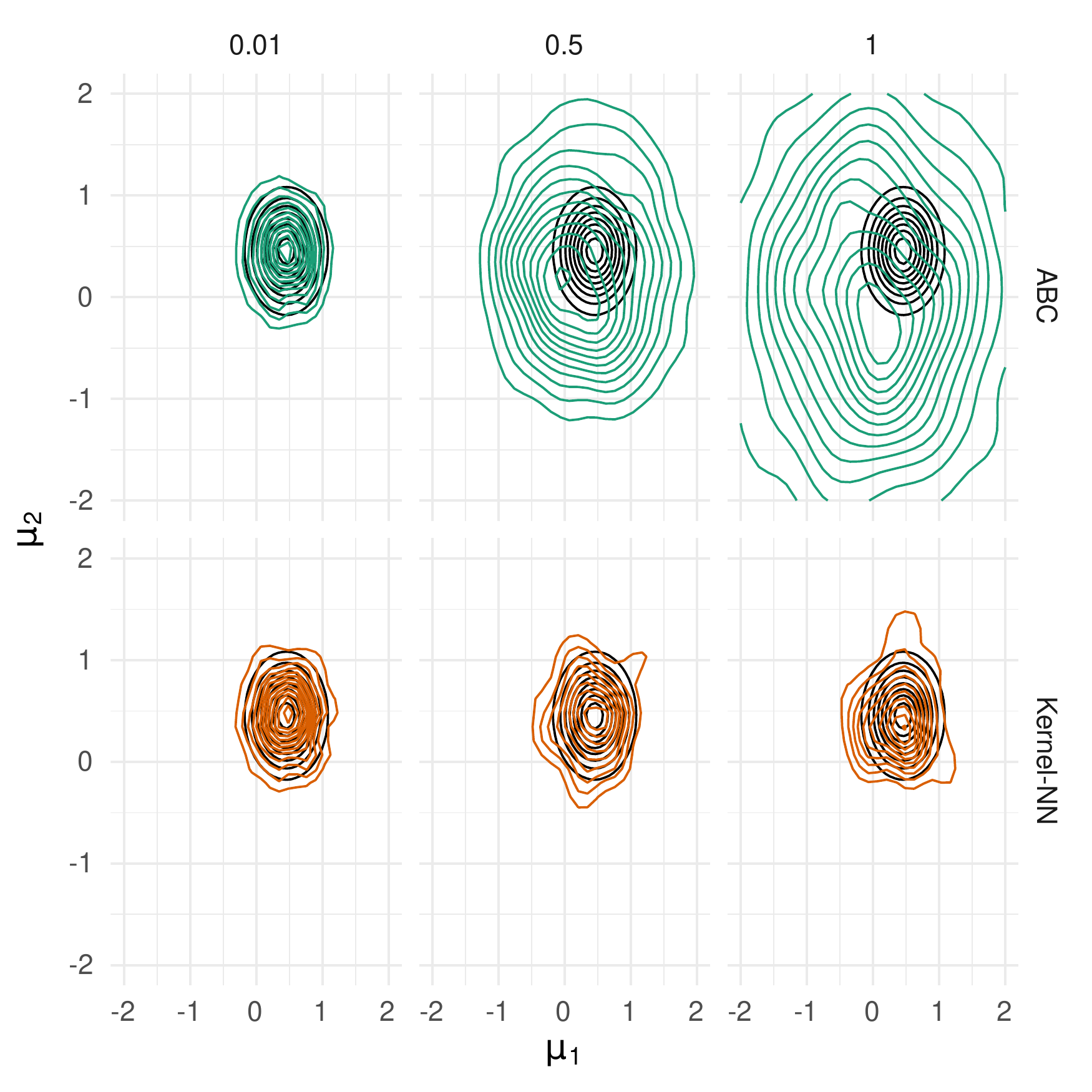}
\caption{\footnotesize Contours for density estimates of two-dimensional normal posterior; black lines are the contours of the true posterior.}
\label{fig:org2bc0769}
\end{figure}

\section{Comparison with ABC Post-Processing Methods}
\label{sec:post-processing}

We can view the connection between ABC and NN-KCDE in yet another way:
NN-KCDE (tuned with the surrogate loss) provides a post-processing step for improving the density
estimate obtained from ABC. There are several other post-processing
procedures in the ABC literature, most notably the regression
adjustment methods of \cite{beaumont2002approximate} and
\cite{blum2010non}; see \citet{li2017} for asymptotic results. These two methods use a regression adjustment to
correct for the impact of the conditional distribution changing with $\x$,
modeling the response as
$$ \theta_{i} = m(\x_{i}) + \sigma(\x_{i}) \times \epsilon_{i}, $$
where $m(\x_{i})$ is the conditional expectation, and $\sigma(\x_{i})$ is
the conditional standard deviation. Assuming this is the true model,
the sample points can be transformed to
\begin{equation}
\label{eq::regression-adjust}  \widetilde{\theta}_{i} = m(\x_{o}) + (\theta_{i} - m(\x_{i})) \times \frac{\sigma(\x_{o})}{\sigma(\x_{i})},
\end{equation}
which scales the sample to have the same mean and standard deviation
as the fitted distribution around $\x_{o}$.

 A visual representation of this procedure can be seen in Figure
\ref{fig::unimodal-transformation}. The data are drawn from 
the conjugate normal posterior described in Section 4.1 of the paper.
In the joint distribution we see a clear linear
relationship between the summary statistic $\bar{\x}$ and parameter
$\theta$. 
The ABC joint sample we obtain from restricting our data to a neighborhoood around $\bar{\x}_{\text{obs}}$ is skewed, as seen in
 the ABC kernel density estimate.
 With access to the true $m(\x)$ and
$\sigma(\x)$, however, one could transform the ABC joint sample with Equation
\ref{eq::regression-adjust} to remove the trend (see regression-adjusted joint distribution) and achieve a better
fit (see regression-adjusted kernel density estimate). \cite{beaumont2002approximate} and \cite{blum2010non} use
local-linear and neural-net regression respectively to estimate
$\widehat{m}(\x)$ and $\widehat{\sigma}(\x)$ with similar effect.

In Figure
\ref{fig::unimodal-losses}, we calculate the ABC
density loss and CDE loss for the unimodal example, and see that regression-adjusted methods achieve similar
performance to ABC-CDE (FlexCode-NN and NN-KCDE) here. We use the abc
package \citep{abcPackage} 
 to fit both the methods of
\cite{beaumont2002approximate} and \cite{blum2010non}.

\begin{figure}[H]
  \centering
 \includegraphics[width=0.6\textwidth]{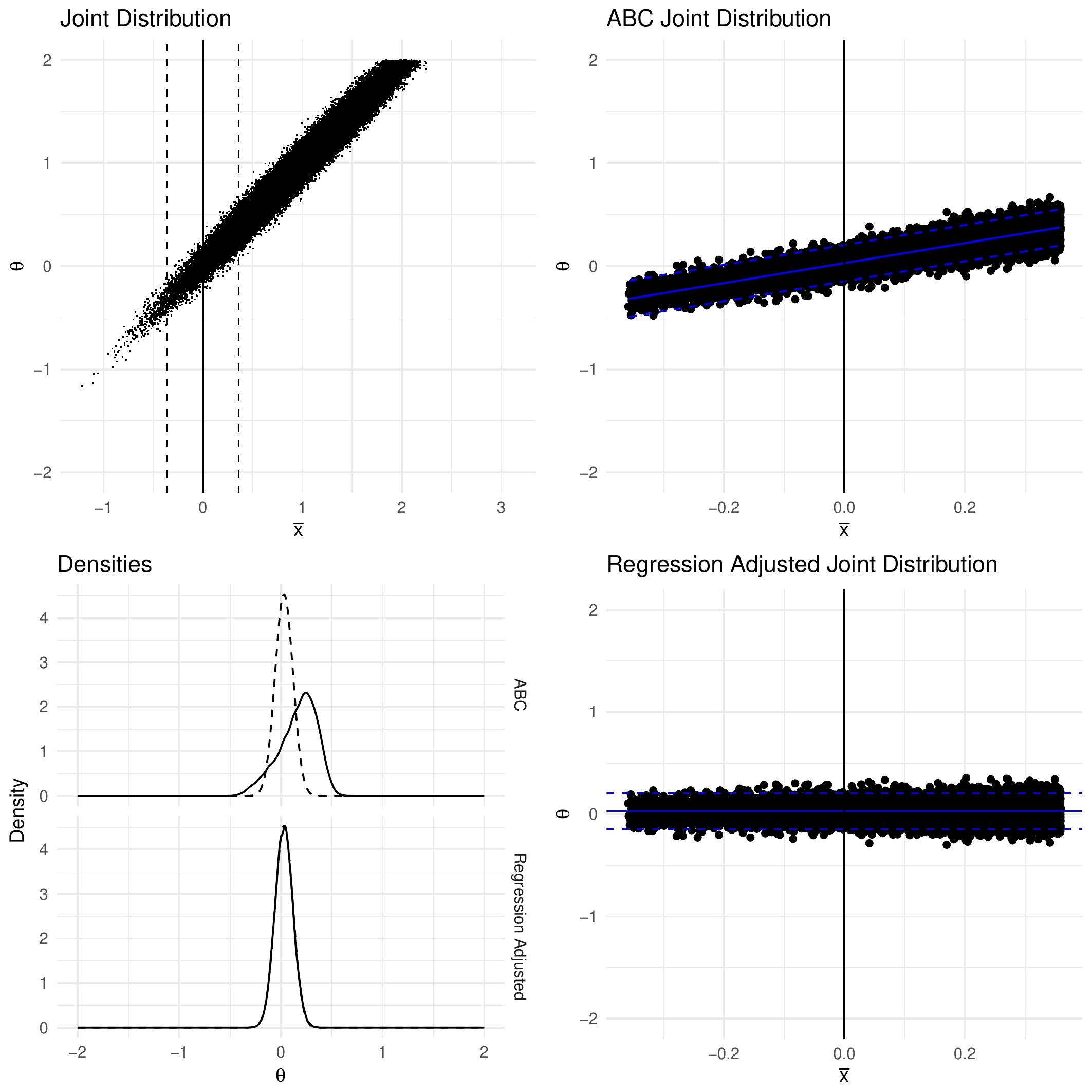}
  \caption{\footnotesize Regression adjustment in ABC for unimodal example assuming we know the true $m(\x)$ and $\sigma(\x)$. In the Joint Distribution plots the solid line represents the conditional mean and the dotted lines are one conditional standard deviation from the mean. See text for details. }
  \label{fig::unimodal-transformation}
\end{figure}


\begin{figure}[H]
  \centering
 \includegraphics[width=0.5\textwidth]{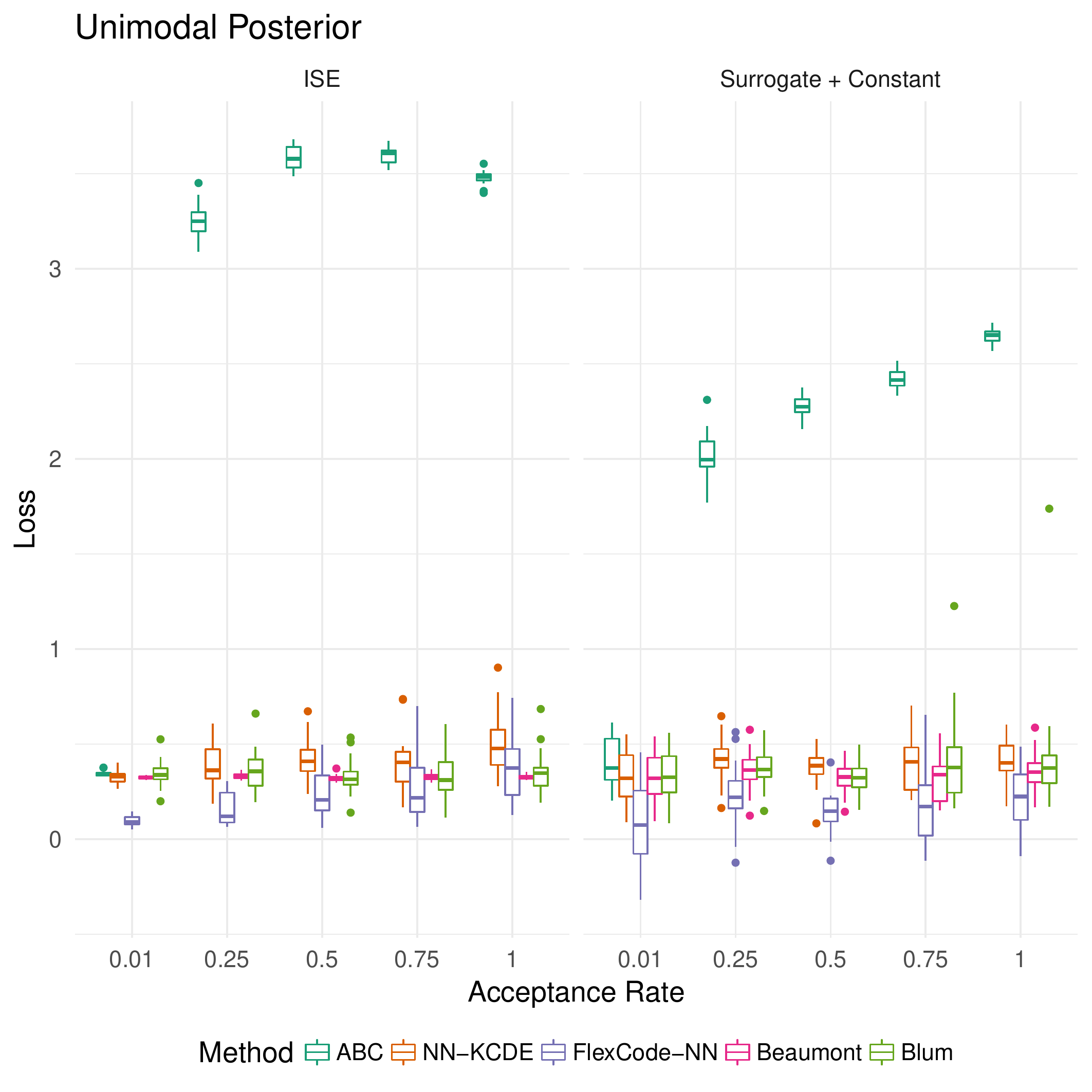}
  \caption{\footnotesize True integrated squared error loss ({\em left}) and estimated surrogate loss  ({\em right})  for unimodal example using ABC  sample of varying acceptance rates. Both regression adjustment and ABC-CDE methods improve upon standard ABC.}
  \label{fig::unimodal-losses}
\end{figure}

However, the regression-based methods rely upon the assumption that
the distributions of $\epsilon_{i}$ at different $\x$ are similar up to a translation and scaling. To illustrate
where this assumption can break down consider the case of a multimodal
posterior. Given the mixture prior
$\mu \sim \sum_{i} w_{i}  \operatorname{Normal}(\mu_{i}, \sigma_{i}^2)$
and likelihood $X_{i} \sim \operatorname{Normal}(\mu, \sigma_{x}^2)$, we  obtain the conjugate mixture model posterior
\begin{equation*}
  \mu \mid X \sim \sum_{i} w_{i}^{*} \operatorname{Normal}(\mu_{i}^{*}, \sigma_{i}^{*2}),
\end{equation*}
where $\mu_{i}^{*}$ and $\sigma_{i}^{*}$ are the parameters of the
conjugate posterior for that particular normal prior. The mixing weights $w_{i}^{*} \propto w_{i}P_{i}(X)$ where $P_{i}(X)$ is the marginal likelihood under the $i$-th mixture component.

We follow the same procedure as Figure
\ref{fig::unimodal-transformation} for this setting resulting in
Figure \ref{fig::multimodal-transformation}. Here the regression is
misleading; the error distribution at $\bar{\x}_{\text{obs}}$ is
bimodal whereas the error distribution away from $\bar{\x}_{\text{obs}}$
is unimodal. When the regression adjustment is made, the bimodality is
lost and a single peak is fit. When we calculate the losses for this
simulation (Figure~\ref{fig::bimodal-losses}), we see that the regression methods perform worse than ABC
whereas the ABC-CDE methods still improve upon ABC. 

\begin{figure}[H]
  \centering
 \includegraphics[width=0.56\textwidth]{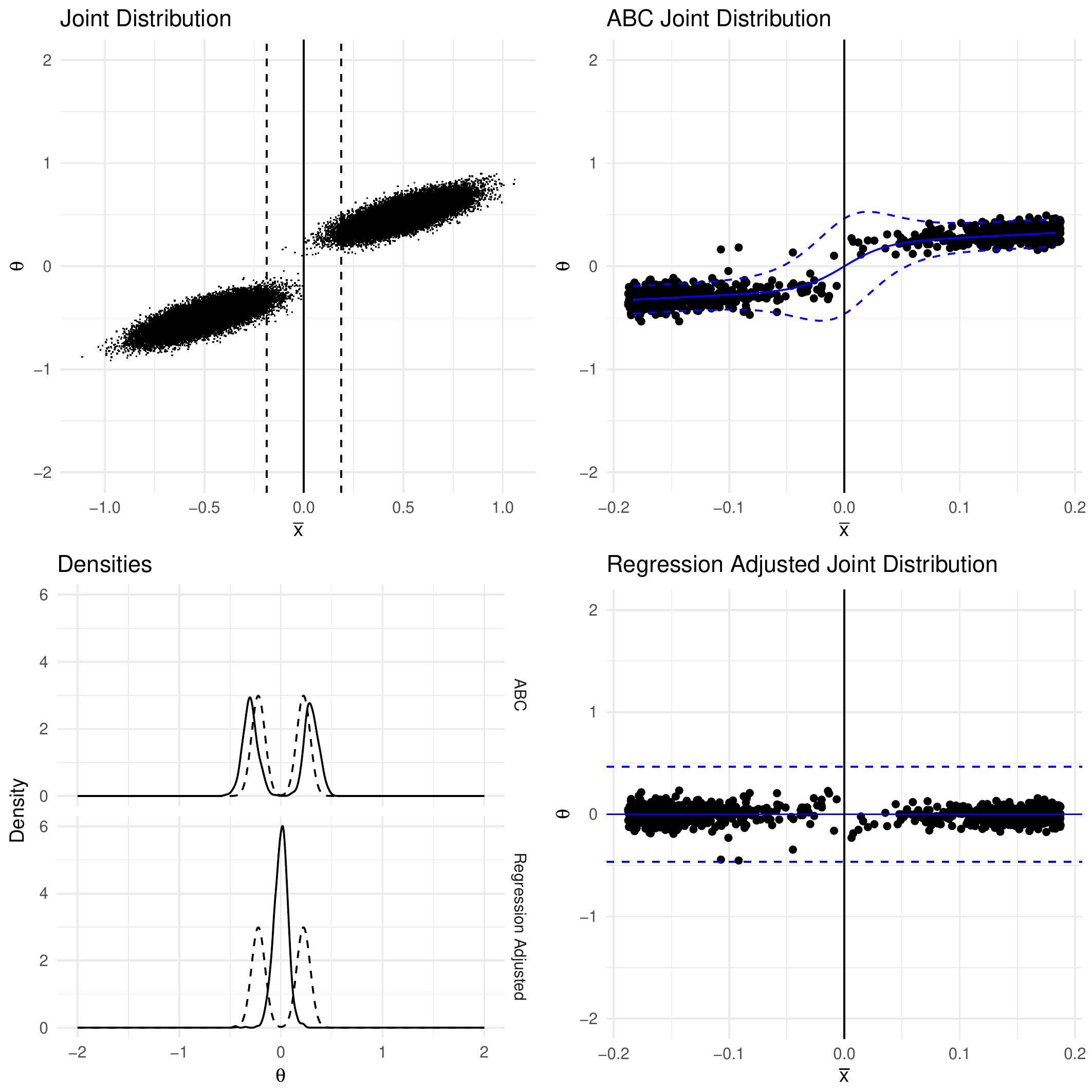}
    \caption{\footnotesize For the multimodal example, a regression of the ABC joint sample is misleading, as we cannot simply adjust for the change in the distribution of $\theta | \x$ around $\x_{\text{obs}}$ by shifting and rescaling the sample by the conditional mean $m(\x)$ and the conditional standard deviation $\sigma(\x)$, respectively.}
  \label{fig::multimodal-transformation}
\end{figure}

\begin{figure}[H]
  \centering
  \includegraphics[width=0.5\textwidth]{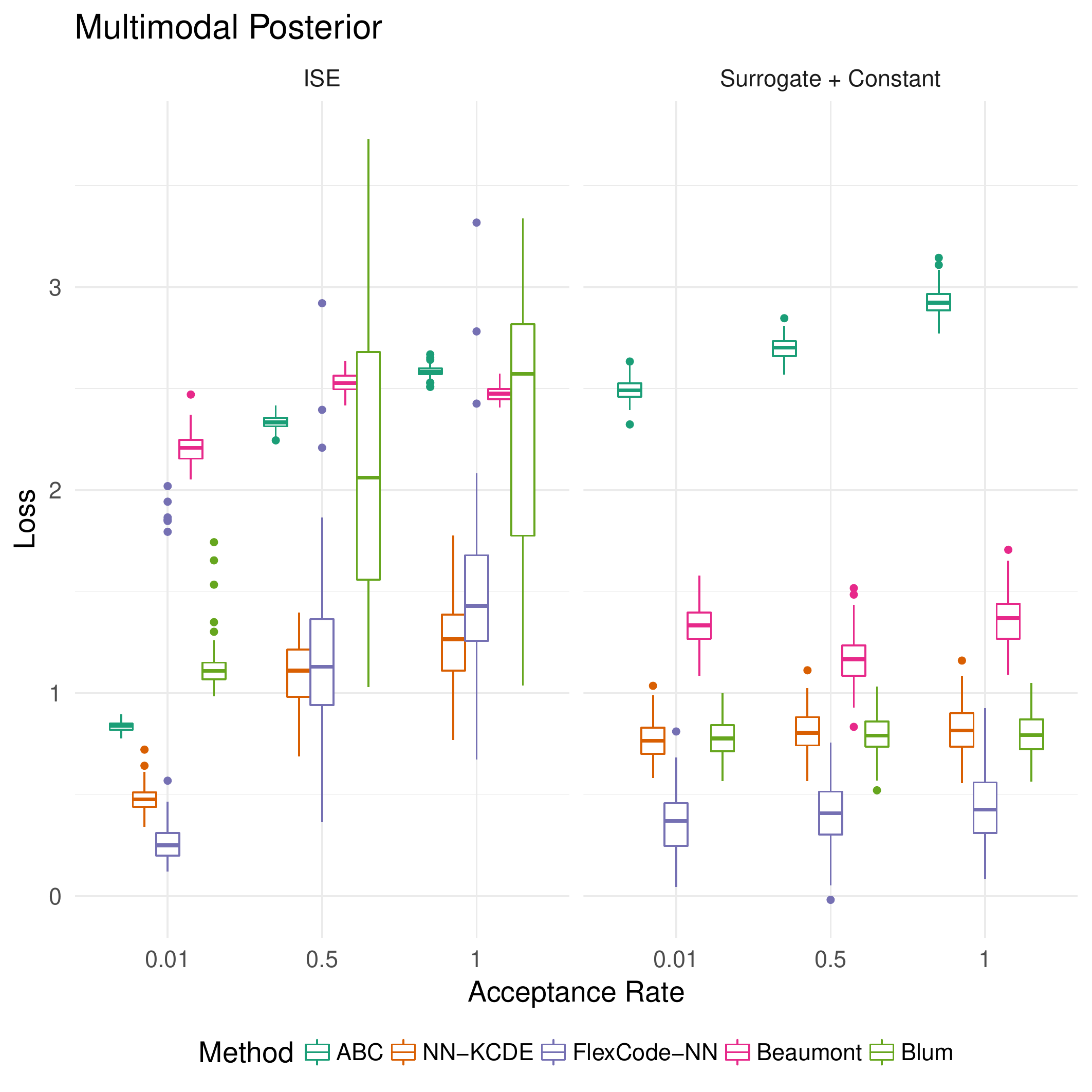}
  \caption{\footnotesize  True integrated squared error loss ({\em left}) and estimated surrogate loss  ({\em right})  for multimodal example using ABC  sample of varying acceptance rates. ABC-CDE methods (e.g, NN-KCDE and FlexCode-NN) improve upon standard ABC, whereas regression adjustment methods (e.g., Beaumont and Blum) can lead to worse results.}
  \label{fig::bimodal-losses}
\end{figure}

\begin{figure}[H]
  \centering
  \includegraphics[scale=0.4]{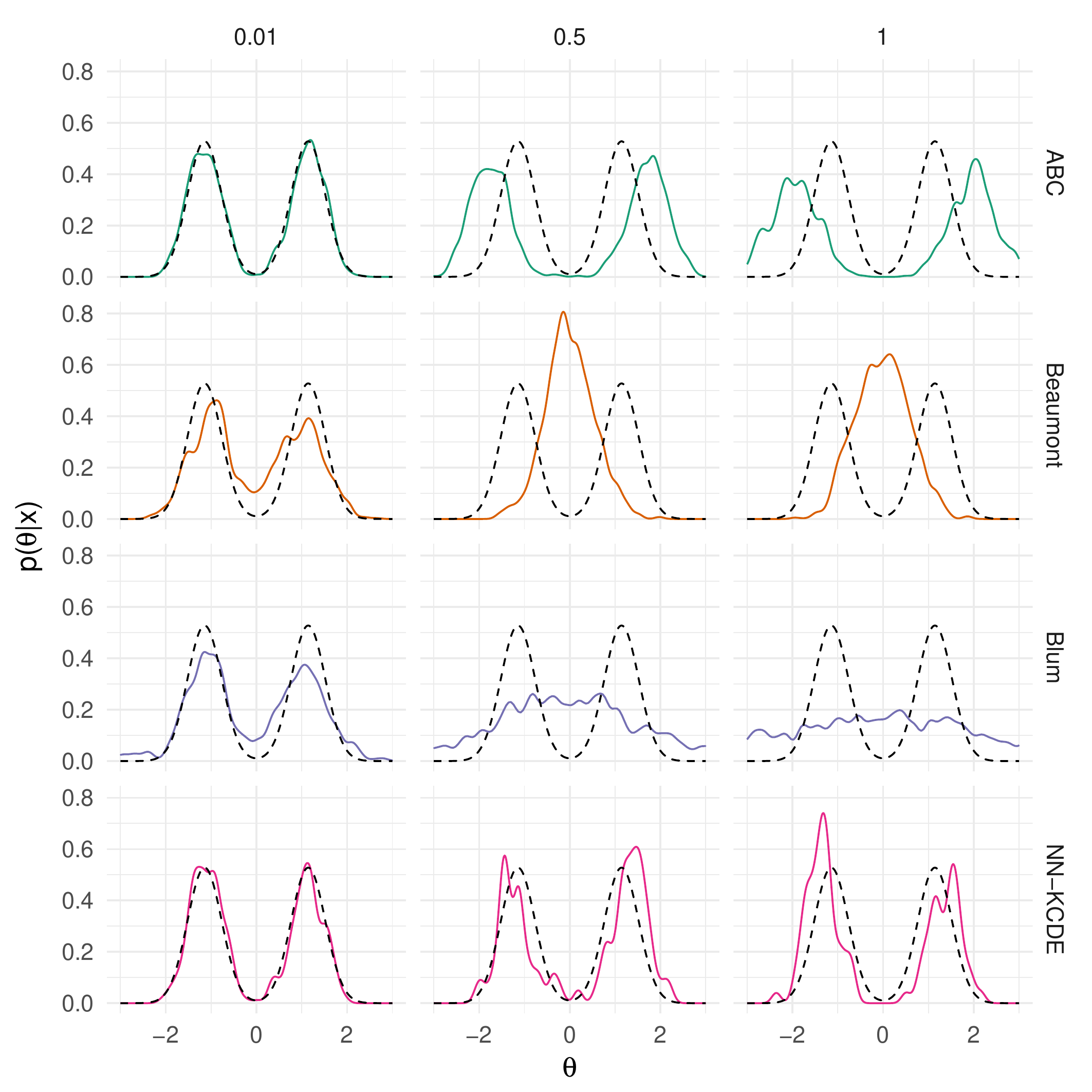}
  \caption{ \footnotesize   The regression-adjustment methods adjust for the change in the distribution of $\theta \mid \x$ around $\x_{\text{obs}}$ by shifting and rescaling the sample by the conditional mean $m(\x)$ and the conditional variance $\sigma(\x)$, respectively. However, the change in the distribution from unimodal to multimodal cannot be expressed by shifting or rescaling which results in misleading posteriors for the regression-adjustment methods for larger values of the ABC tolerance while NN-KCDE performs well.}
  \label{fig::multimodal-densities}
\end{figure}

\section{Application II: Cosmological Parameter Inference via Weak Lensing}
\label{sec:org9e43ec4}
We end by considering the problem of cosmological parameter inference via weak gravitational lensing.  
Gravitational lensing causes distortion in images of distance
galaxies; this is called cosmic shear. Because the universe has
varying matter densities, these create tidal gravitational fields which
cause light to deflect differentially. The size and direction of
distortion is directly related to the size and shape of the matter
along that line of sight. We can use shear correlation functions
to study the properties and evolution of the large scale structure
and geometry of the Universe. In particular we can constrain
parameters of the \(\Lambda CDM\) cosmological model such as the dark
matter density \(\Omega_{M}\) and matter power spectrum normalization \(\sigma_{8}\).
For further background see \cite{hoekstra2008weak}, \cite{munshi2008cosmology} and \citet{mandelbaum2017weak}.

We can perform ABC rejection sampling using the Euclidean distance
between binned shear correlation functions as our summary statistic.
We use the \texttt{lenstools} package \citep{petri2016mocking} to
generate power spectra given parameter realizations. We then use the
\texttt{GalSim} toolkit \citep{rowe2015galsim} to generate simplified
galaxy shears distributed according to a Gaussian random field
determined by $(\Omega_{M}, \sigma_{8})$.

For our prior distribution we assume a uniform distribution:
\(\Omega_{M} \sim U(0.1, 0.8)\) and \(\sigma_{8} \sim U(0.5, 1.0)\).
Other parameters are fixed to \(h = 0.7\),
\(\Omega_{b} = 0.045\), \(z = 0.7\).
\label{org0eea9fd}

\begin{figure}[htbp]
\centering
	\includegraphics[width=0.5\textwidth]{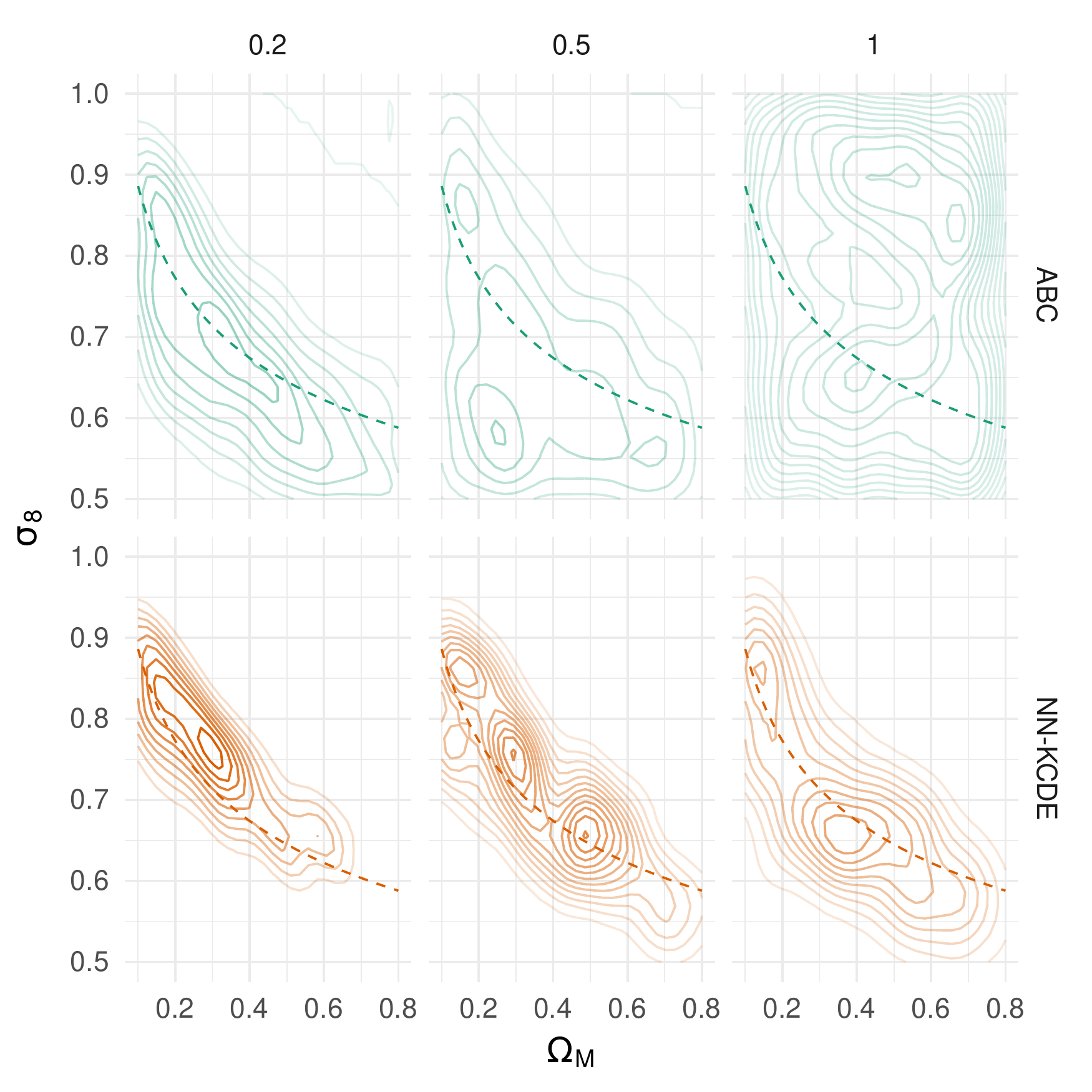}
\caption{\label{fig:org1743082}
\footnotesize Estimated posteriors of cosmological parameters for weak lensing mock data at different ABC acceptance rates (0.2, 0.5, and 1). The dashed line represents the parameter degeneracy curve on which the data are indistinguishable. With  NN-KCDE tuned with a surrogate loss, the posteriors concentrate rapidly around the degeneracy line; we even see some structure for an ABC acceptance rate of 1; that is, an ABC threshold of $\epsilon \rightarrow \infty$.}
\end{figure}

One result that is apparent from Figure~\ref{fig:org1743082} is that kernel-NN tuned with the surrogate loss quickly converges to 
the degeneracy curve \(\Omega_{M}^{\alpha}\sigma_{8}\) on which observable data are indistinguishable.  
 
 As we in the future analyze more complex simulation mechanisms and higher-dimensional data (with, for example, galaxies divided into time-space bins and measurements from different probes), the dimension of the data and the simulation time will eventually make standard ABC intractable. For example, recent analyses in cosmological analysis \citep[e.g.,][]{DES2016,KIDS2017} have employed $\sim$1000 expensive N-body simulations, which altogether can take months to run on thousands of CPUs \citep{Masanori2011,Deraps2015}. Often there are also several parameters, which results in prior distributions
that are typically not concentrated around the true value of $\theta$. Our work indicates that these are exactly the settings where
nonparametric conditional density estimators of $f(\theta|\x)$
will lead to better performance than ABC.

\bibliographystyle{Chicago}

{ \footnotesize 
\bibliography{paper}
}

\end{document}